\documentclass[12pt,reqno]{amsart}
\usepackage{amssymb}
\usepackage{amsmath}
\usepackage{mathrsfs}

\usepackage{amssymb, amsmath, amsfonts, latexsym}
\usepackage{enumerate}
\usepackage{graphicx}
\usepackage{subfigure}
\usepackage{float}

\newtheorem{thm}{Theorem}[section]
\newtheorem{cor}[thm]{Corollary}
\newtheorem{lem}[thm]{Lemma}
\newtheorem{prop}[thm]{Proposition}
\newtheorem{example}[thm]{Example}
\newtheorem{observation}[thm]{Observation}
\newtheorem{remarks}[thm]{Remark}
\newtheorem{hypothesis}{Hypothesis}

\theoremstyle{definition}
\newtheorem{defn}{Definition}[section]

\numberwithin{equation}{section} \theoremstyle{remark}

\newcommand{\rr}{\mathbb{R}}

\def\sup{{\mathrm{{\rm sup}}}}

\def\AA{\mathcal A}

\def\<{\langle}
\def\>{\rangle}

\def\beq{\begin{equation}}
\def\nneq{\end{equation}}

\def\bdef{\begin{defn}}
\def\ndef{\end{defn}}

\def\bthm{\begin{thm}}
\def\nthm{\end{thm}}

\def\bprop{\begin{prop}}
\def\nprop{\end{prop}}

\def\brmk{\begin{remarks}}
\def\nrmk{\end{remarks}}

\def\bexa{\begin{example}}
\def\nexa{\end{example}}

\def\obse{\begin{observation}}
\def\eobse{\end{observation}}

\def\blem{\begin{lem}}
\def\nlem{\end{lem}}

\def\bcor{\begin{cor}}
\def\ncor{\end{cor}}

\date{}

\def\bexe{\begin{exe}}
\def\nexe{\end{exe}}

\def\bprf{\begin{proof}}
\def\nprf{\end{proof}}

\def\bdes{\begin{description}}
\def\ndes{\end{description}}

\def\blem{\begin{lem}}
\def\nlem{\end{lem}}

\def\blem{\begin{lem}}
\def\nlem{\end{lem}}

\author{Nian Yao}

\address{Nian Yao\\School of Mathematics and Statistics, Shenzhen University, Shenzhen, Guangdong Province, 518060, China.}
\email{yaonian@szu.edu.cn}

\author{Zhiming Yang }

\address{Zhiming Yang\\School of Mathematics and Statistics, Shenzhen University, Shenzhen, Guangdong Province, 518060, China. }
\email{yzm1991317@qq.com}

\begin{document}
\date{}
\title[Optimal excess-of-loss reinsurance and investment problem for an
insurer with default risk under a stochastic volatility model]{Optimal excess-of-loss
reinsurance and investment problem for an insurer with default risk under a
stochastic volatility model}
\maketitle

\noindent \textbf{Abstract:} In this paper, we study an optimal
excess-of-loss reinsurance and investment problem for an insurer in
defaultable market. The insurer can buy reinsurance and invest in the
following securities: a bank account, a risky asset with stochastic
volatility and a defaultable corporate bond. We discuss the optimal
investment strategy into two subproblems: a pre-default case and a
post-default case. We show the existence of a classical solution to a
pre-default case via super-sub solution techniques and give an explicit
characterization of the optimal reinsurance and investment policies that
maximize the expected CARA utility of the terminal wealth. We prove a
verification theorem establishing the uniqueness of the solution. Numerical
results are presented in the case of the Scott model and we discuss economic
insights obtained from these results. \vskip0.3cm

 \vskip0.3cm


\section{Introduction}

\noindent\ \ The theory of optimal investment dates back to the seminal
works of Merton (1969, 1971, 1990). In the setting of continuous-time
models, an optimization problem of an agent who invests his/her wealth into
a financial market to maximize the expected utility of terminal wealth was
studied. He derived a solution to this optimization problem for a complete
market by employing tools of optimal stochastic control. Browne(1995)
considered the risk process is approximated by a Browmian motion with drift
and the stock price process modeled by a geometric Browmian motion and the
insurer maximizes the expected constant absolute risk aversion(CARA) utility
from the terminal wealth. Under this assumption, when the interest rate of a
risk-free bond is zero, the optimal strategy also minimizes the ruin
probability. Hipp and Plum(2000) studied risk process follows the classical
Cramer-Lundb$\acute{e}$rg model and the insurer can invest in a risky asset
to minimize the ruin probability. However, the interest rate of the bond in
their model is implicitly assumed to be zero. Liu and Yang(2004) extended
the model of hipp and Plum(2000) to incorporate a non-zero interest rate.
But in this case ,a closed-form solution cannot be obtained. Yang and
Zhang(2005)considered that the insurer is allowed to invest in the money
market and a risky asset. They obtained a closed form expression of the
optimal strategy when the utility function is exponential. Fern$\acute{a}$%
ndez et al.(2008) considered the risk model with the possibility of
investment in the money market and a risky asset modeled by a geometric
Brownian motion. Via the Hamilton-Jacobi-Bellman(HJB) approach, they found
the optimal strategy when the insurer's preferences are exponential.
Badaoui(2013) extended the model of Fern$\acute{a}$ndez et al.(2008) to a
risky asset with stochastic volatility, when the insurer preferences are
exponential, they prove the existence of a smooth solution, and they give an
explicit form of the optimal strategy. \vskip0.3cm

For the reinsurance problem, Promislow and Young (2005) obtained investment
and reinsurance strategies to minimize the ruin probability for a diffusion
risk model. Bai and Guo (2008) considered an optimal proportional
reinsurance and investment problem with multiple risky assets for a
diffusion risk model. Cao and Wan (2009) investigated the proportional
reinsurance and investment problem of utility maximization for an insurance
company. Zeng and Li (2011) obtained the time-consistent investment and
proportional reinsurance policies under the mean-variance criterion for an
insurer. Gu et al. (2010) introduced the CEV model into the optimal
reinsurance and investment problem for insurers. Later, Liang et al. (2012)
and Lin and Li (2011) investigated the optimal reinsurance and investment
problem for an insurer with a jump diffusion risk process under the CEV
model. Li et al. (2012) began to apply the Heston model to study the
reinsurance and investment problem under the mean-variance criterion.
Asmussen et al. (2000) firstly studied the optimal dividend problem under
the control of excess-of-loss reinsurance and showed that excess-of-loss
reinsurance is more profitable than the proportional reinsurance. Zhao and
Rong(2013) considered the risk process approximated by a Heston model with
drift and they obtained the optimal excess-of-loss reinsurance strategy. %
\vskip0.3cm

For the risk of default problem, Bielecki and Jang(2006) considered that the
insurer is allowed to invest in bond and risky asset and default asset whose
coefficient is constant. Capponi and Figueroa-L$\acute{o}$pez(2014)
considered the same problem that the risky asset is a markov process with
multi-dimensional continuous time in finite state. In these two articles,
the dynamic programming method was adopted, and the optimal strategy was
obtained. Jiao and Pham (2011) used a default-density modelling approach and
addressed the power utility maximization problem using the terminal wealth
in a financial market with a stock exposed to a counter-party risk. By
decomposing the optimization problem into two sub-problems, one that is
stated before the default time and one that is stated after default, they
derive the optimal investment strategy by applying standard martingale
approaches. Bo et al. (2010, 2013) considered a portfolio optimization
problem with default risk under the intensity-based reduced-form framework,
and the goal was to maximize the infinite horizon expected discounted HARA
utility of consumption, where the default risk premium and the default
intensity were assumed to rely on a stochastic factor described by a
diffusion process. Zhu et al.(2015) studied the optimal investment and
reinsurance problem for an insurer whose investment opportunity set contains
a default security and the closed-form expressions for optimal control
strategies and the corresponding value functions are derived. Bo et al.
(2016) considered an optimal risk-sensitive portfolio allocation problem,
which explicitly accounts for the interaction between market and credit risk
and show the existence of a classical solution to this system via super-sub
solution techniques and give an explicit characterization of the optimal
feedback strategy. \vskip0.3cm

In our paper, the insurer is allowed to purchase excess-of-loss reinsurance
and invest in a risk-free asset and a risky stock asset follows the general
stochastic volatility model and a defaultable corporate bond. Comparing with
Badaoui(2013) and Zhu et al.(2015), we add an excess-of-loss reinsurance and
default risk into the model and generalize the Heston model to the more
general stochastic volatility model. We work under the martingale invarance
hypothesis. Herein, we also assume the existence of the conditional density
of the default time $\tau $ . Let the surplus process of the insurer satisfy
a jump--diffusion process, and the dynamics of the risky stock price follow
a stochastic volatility model. The insurance company's manager can
dynamically choose a proportion reinsurance strategy and allocate the wealth
into the above three assets. The goal is to maximize the finite horizon
expected exponential utility of terminal wealth. In the spirit of Bielecki
and Jang(2006), we decompose the original optimization problem into two
sub-problems: a pre-default case and a post-default case. A dynamic
programming principle is employed to derive the Hamilton--Jacobi--Bellman
(HJB) equation. We show the existence of a classical solution to a
pre-default case via super-sub solution techiniques. The closed-form
expressions for optimal control strategies and the corresponding value
functions are derived.

The remainder of this paper is organized as follows: In Section 2, we
introduce the model and the problem of our research.  In Section 3, we
derive the HJB equation for the pre-default case and the post-default case,
and then, the explicit expressions for optimal control strategies and the
corresponding value functions are obtained. And also we show the existence
of a classical solution to a pre-default case via super-sub solution
techiniques. In addition, we provide the verification theorem. In Section 4
demonstrates our results with numerical examples. \vskip0.3cm

In the Appendix we give some results about Partial Differential Equations
which is important to our proof.

\section{The model}

\subsection{Dynamics of reserve process}

The insurer's surplus process is described by the classical risk model
perturbed by a diffusion, i.e.,
\begin{equation}
dR_t=cdt-dC_t,
\end{equation}
where $c$ is the premium rate, $C_t$ represents the cumulative claims up to
time $t$. Suppose the premium is
calculated according to the expected value principle, i.e., $c=(1+\eta
)\lambda\mu_{\infty}$, where $\eta>0$ is the safety loading of the insurer.
We assume that $C_t=\sum\limits_{i=1}^{N_{t}}X_{i}$ is a compound Poisson
process, where $N_{t}$ is a homogeneous Poisson process with intensity $%
\lambda$ and jump times $\{T_{i}\}_{i\geq1}$. The claim sizes $%
\{X_{i},i\geq1\}$ are independent and identically distributed positive
random variables with common distribution $F(x)$. Denote the mean value $%
E[X_{i}]=\mu_{\infty}$ and $D:={\mathrm{\mathrm{sup}}}\{z: F(z)<+\infty\}$.
Suppose that $F(0)=0$, $0<F(x)<1$ for $0<x<D$ and $F(x)=1$ for $x\geq D$. In
addition, we assume that $N_{t}$ is independent of the claim sizes $X_{i}$, $%
i\geq1$. \newline

The insurer is allowed to purchase excess-of-loss reinsurance to reduce the
underlying insurance risk. Let $a$ be a (fixed) excess-of loss retention
level. Then the corresponding reserve process is
\begin{equation}
dR_{t}=c^{(a)}dt-dC^{(a)}_{t},
\end{equation}
where\newline
\begin{align*}
c^{(a)} &
=(1+\eta)\lambda\mu_{\infty}-(1+\theta)\lambda\{\mu_{\infty}-E[%
\min(X_{1},a)]\} \\
& =(\eta-\theta)\lambda\mu_{\infty}+(1+\theta)\lambda\int_{0}^{a} \bar {F}%
(x)dx,
\end{align*}
$C^{(a)}_{t}=\sum\limits_{i=1}^{N_{t}}\min(X_{i},a)$ and $\theta$ denotes
the safety loading of the reinsurer and $\bar{F}(x)=1-F(x)$. Without loss of
generality, we assume that $\theta>\eta$ and
\begin{equation*}
\exp\left\{ \int_{0}^{t} e^{-rs}dC^{(a)}(s)\right\} <\infty, \forall
t<\infty.
\end{equation*}

\subsection{The financial market}

We assume $(\Omega,\mathcal{G},\mathbb{Q})$ to be a complete probability
space that is endowed with a reference filtration $\mathcal{F}=\{\mathcal{F}%
_{t}\}_{t\geq0}$ that satisfies the usual conditions. The probability
measure $\mathbb{Q}$ is a martingale probability measure and is assumed to
be equivalent to the real-world measure $\mathbb{P}$. Let $\tau$ be a
non-negative random variable on this space. $\tau$ represents the first jump
time of a Poisson process with constant intensity $h^{Q}>0$. For the sake of
convenience, we assume that $\mathbb{Q}(\tau=0)=0$ and $\mathbb{Q}(\tau>0)>0$%
, which implies that the default cannot occur at the initial time and can
occur at any time until maturity. For $t\geq0$, define a default indicator
process $H=(H_{t};t\geq0)$ by $H_{t}=\mathbb{I}_{\{\tau\leq t\}}$. The
filtration $\mathcal{G}$ is defined using $\mathcal{G}_{t}=\mathcal{F}%
_{t}\bigvee \sigma(H(t);s\leq t)=\mathcal{F}_{t}\bigvee\sigma(\tau\bigwedge
t)$. Then, $\mathcal{G}=(\mathcal{G}_{t};t\geq0)$ is the smallest filtration
such that the random time $\tau$ is not necessarily a stopping time, and $%
\mathcal{G}_{t}$ is called the enlarged filtration. Such an information
structure is standard in the reduced-form approach.

Let the conditional survival probability be given by
\begin{equation}
\mathbb{Q}(\tau>t|\mathcal{F})=e^{-h^{Q}t},
\end{equation}
where the risk neutral intensity $h^{Q}$ is assumed to be constant; then,
the following process related to default
\begin{equation}
M^{Q}_{t}=H_{t}-\int_{0}^{t} (1-H_{u})h^{Q}du,
\end{equation}
is a $(\mathbb{Q},\mathcal{G})$ martingale.

By applying Proposition 1 in Zhu(2015), the $\mathbb{P}$-dynamics of the
defaultable bond price process $p(t,T_{1})$ are given by
\begin{equation}
dp(t,T_{1})=p(t-,T_{1})[r(Z_{t})dt+(1-H_{t})\delta(1-\Delta)dt-(1-H_{t-})%
\zeta dM^{P}_{t}],
\end{equation}
where $M^{P}_{t}=H_{t}-h^{Q}\int_{0}^{t} (1-H_{u})\Delta du$ is a $\mathcal{G%
}$-martingale under the real-world probability $\mathbb{P}$ and $%
\delta=h^{Q}\zeta$ is the credit spread under the real-world probability
measure, $\zeta$ is the loss rate, $h^{P}=h^{Q}\Delta$ is a constant and $%
\frac{1}{\Delta}\geq1$ denote the default risk premium.

The price process of the risk-free asset is given by
\begin{equation}
dS_{t}^{0}=S_{t}^{0}r(Z_{t})dt,
\end{equation}
where $r(\cdot)$ is the interest rate function. The process $Z_{t}$ can be
interpreted as the behavior of some economic factor that has an impact on
the dynamics of the risky asset and the bank account. For instance, the
external factor can be modeled by the mean reverting Ornstein-Uhlenbeck
(O-U) process:
\begin{equation}
dZ_{t}=\delta(\kappa-Z_{t})dt+\beta d\widetilde{W}_{t},Z_{0}=z,
\end{equation}
where $\delta$ and $\kappa$ are constant.

From Badaoui(2013), we assume the risky asset price satisfies the following
stochastic volatility model:
\begin{equation}  \label{price equation}
dS_{t}=S_{t}(\mu(Z_{t})dt+\sigma(Z_{t})dW_{1t}),
\end{equation}
where $S_{0}=1$, $W_{1t}$ is a standard Brownian motion; $\mu(\cdot)$ and $%
\sigma(\cdot)$ are respectively the return rate and volatility functions. $Z$
is an external factor modeled as a diffusion process solving

\begin{equation}  \label{SV equation}
dZ_{t}=g(Z_{t})dt+\beta(\rho dW_{1t}+\sqrt{1-\rho^{2}}dW_{2t}),
\end{equation}
where $Z_{0}=z\in\mathbb{R}$, $|\rho|\leq1$ and $\beta\neq0$, $W_{2t}$ is a
standard Brownian motion , $W_{1t}$ and $W_{2t}$ are independent and $%
\widetilde{W}=\rho W_{1t}+\sqrt{1-\rho^{2}}W_{2t}$. For example the risky asset price can be
given by the Scott model (Fouque et al., 2000; Rama and Peter, 2003):
\begin{equation}
dS_{t}=S_{t}(\mu_{0}dt+e^{Z_{t}}dW_{1t}),S_{0}=1,
\end{equation}
Here, we assume that $\mu_{0}$ is constant.

More details about stochastic volatility models can be bound in Fouque et
al. (2000).

\subsection{The wealth process}

We assume that the insurer is allowed to purchase excess-of-loss
reinsurance. The insurer has investment opportunities in a risky stock
asset, a risk-free asset and a corporate bond issued by a private
corporation, which may default at some random time $\tau$ , where the
investment horizon is $[0,T]$ and $T<T_{1}$. Let $\pi(t)=(l(t), m(t), a(t))$
be the reinsurance-investment strategy followed by the insurer, where $l(t)$
represents the amount of wealth invested into the stock market, $m(t)$ is
the amount of wealth invested in the corporate bond, and $a(t)$ denotes the
reinsurance strategy at time $t$. We assume that the corporate bond is not
traded after default. Let $\mathcal{A}$ denote all admissible strategies.
The reserve process subjected to this choice is denoted by $%
Y^{\pi}_{t}=Y(t,y,z,\pi)$, and its dynamics are given by %
\setcounter{equation}{10}
\begin{equation}  \label{risk process}
\begin{split}
dY^{\pi}_{t} & =\frac{(Y^{\pi}_{t}-l(t)-m(t))}{S_{t}^{0}}dS_{t}^{0}+\frac{%
l(t)}{S_{t}}dS_{t}+\frac{m(t)}{p(t)}dp(t)+dR_{t} \\
&
=[r(Z_{t})Y^{\pi}_{t}+(\mu(Z_{t})-r(Z_{t}))l(t)+c^{(a)}+(1-H_{t})m(t)%
\delta(1-\triangle)]dt \\
& +l(t)\sigma(Z_{t})dW_{1t}-m(t)(1-H_{t})\zeta
dM^{P}_{t}-d\sum_{i=1}^{N_{t}}\min(X_{i},a(t)).
\end{split}%
\end{equation}
Suppose that the insurer is interested in maximizing the CARA utility
function for his terminal wealth, say, at time $T$. The utility function is $%
U(y)=-e^{-\alpha y}$, $\alpha>0$, which is satisfies $U^{^{\prime}}>0$ and $%
U^{^{\prime\prime}}<0$. We are now in a position to formulate the following
optimization problem:
\begin{equation}
V(t,y,z,h)=\mathop{\sup}\limits_{\pi\in\mathcal{A}}
E^{P}[U(Y^{\pi}_{T})|(Y^{\pi}_{t},Z_{t},H_{t})=(y,z,h)].
\end{equation}

\begin{hypothesis}
\label{coefficient} 1.The functions $\mu(\cdot)$, $\sigma(\cdot)$ and $%
g(\cdot)$ are such that there exists a strong solution for Eqs.(\ref{price
equation}) and (\ref{SV equation}).\newline
2.The function $r(\cdot)$ is continuous, positive, and $r(z)<\mu(z)$, for
all $z\in \mathbb{R}$.\newline
\end{hypothesis}

\section{The main result}

Using dynamic programming techniques ,we find the corresponding HJB equation
is
\begin{equation}  \label{HJB}
\left\{ \begin{aligned} &\mathop{\sup}\limits_{\pi\in\AA}\mathcal
L^{\pi}J(t,y,z,h)=0,\\ &J(T,y,z,h)=U(y). \end{aligned}\right.
\end{equation}

where
\begin{equation}
\begin{split}
\mathcal{L}^{\pi}J(t,y,z,h)= & J_{t}(t,y,z,h)+J_{y}(t,y,z,h)\bigg(%
r(z)y+l(t)(\mu(z)-r(z))+c^{(a)}+m(t)(1-h)\delta\bigg) \\
& +J_{z}(t,y,z,h)g(z)+\frac{1}{2}J_{yy}(t,y,z,h)l(t)^{2}\sigma(z)^{2}+\frac {%
1}{2}J_{zz}(t,y,z,h)\beta^{2} \\
& +J_{yz}(t,y,z,h)\beta\rho\sigma(z)l(t)+\lambda\bigg(EJ(t,y-\min
(X_{1},a),z,h)-EJ(t,y,z,h)\bigg) \\
& +\bigg(J(t,y-m(t)\zeta,z,h+1)-J(t,y,z,h)\bigg)h^{P}(1-h).
\end{split}%
\end{equation}
Now we establish a verification theorem, which relates the value function $V$
with the HJB equation (\ref{HJB}).

\begin{thm}
\label{Verification Theorem}(Verification Theorem). Let $J(t,y,z,h)$ with $%
(t,y,z,h)\in[0,T]\times R \times R\times\{0,1\}$ be the classical solution
to the HJB equation (\ref{HJB}) with terminal condition $J(T,y,z,h)=U(y)$
for all $(y,z)\in R^{2}$. Also assume that for each $\pi\in\mathcal{A}$,
\begin{align}
& \int_{0}^{T} \int_{0}^{\infty}\mathbb{E}\left|
J(t,Y_{t}^{\pi}-\min(x,a),Z_{t},H_{t})-J(t,Y_{t-}^{\pi},Z_{t},H_{t})\right|
^{2}dF(x)dt<\infty,  \label{assumption 1} \\
& \int_{0}^{T} \mathbb{E}\left|
l(t,z)J_{y}(t,Y_{t-}^{\pi},Z_{t},H_{t})\right| ^{2}dt<\infty, \int_{0}^{T}
\mathbb{E}\left| J_{z}(t,Y_{t-}^{\pi},Z_{t},H_{t})\right| ^{2}dt<\infty,
\label{assumption 2} \\
& \forall s\in[0,T], \bigg\{\int_{s}^{v}
(J(t,Y_{t}^{\pi}-m(t)\zeta,Z_{t},1-H_{t})-J(t,Y_{t-}^{%
\pi},Z_{t-},H_{t-}))dM_{t}^{P}\bigg\}_{v\in[s,T]} \mathrm{~is~a~martingale} .
\label{assumption 3}
\end{align}
Then, under hypothesis (\ref{coefficient}-\ref{hypothesis 2}) and
assumptation (\ref{assumption 1}-\ref{assumption 3}), for each $u\in[0,t]%
,(y,z)\in R^{2}$,
\begin{equation}  \label{Value-DPE}
J(u,y,z,h)\geq V(u,y,z,h),
\end{equation}
If, in addition, there exists an optimal strategy $\pi^{*}$, then
\begin{equation*}
J(u,y,z,h)=V(u,y,z,h)=E[U(Y_{T}^{\pi^{*}})|(Y^{%
\pi^{*}}_{u},Z_{u},H_{u})=(y,z,h)].
\end{equation*}
\end{thm}

\begin{proof}
We only prove the pre-default case when $h=0$. The default-case $h=1$ is the
same as the pre-default case. Let $\pi\in\mathcal{A}$. Ito's formula implies
that for any $v\in[u,T]$,
\begin{equation}
\begin{aligned}\label{Ito's formula}
J(v,Y_v^{u,y,z,\pi},Z_v,H_v)&=J(u,y,z,0)+\int_u^vJ_t(t,Y_t^{u,y,z,%
\pi},Z_t,H_t)dt+\int_u^vJ_y(t,Y_t^{u,y,z,\pi},Z_t,H_t)dY_t^c\\
&+\int_u^vJ_z(t,Y_t^{u,y,z,\pi},Z_t,H_t)d
Z_t+\frac{1}{2}\int_u^vJ_{yy}(t,Y_t^{u,y,z,\pi},Z_t,H_t)d\<Y_t^c,Y_t^c\>_t\\
&+\frac{1}{2}\int_u^vJ_{zz}(t,Y_t^{u,y,z,\pi},Z_t,H_t)d\<Z,Z\>_t+%
\int_u^vJ_{yz}(t,Y_t^{u,y,z,\pi},Z_t,H_t)d\<Y,Z\>_t\\
&+\int_u^v(J(t,Y_t^{u,y,z,\pi}-m(t)\zeta,Z_t,1-H_t)-J(t,Y_{t-}^{u,y,z,%
\pi},Z_{t-},H_{t-}))d H_t\\ &+\int_u^v\int_0^\infty
(J(t,Y_t^{u,y,z,\pi}-\min(x,a),Z_t,H_t)-J(t,Y_{t-}^{u,y,z,\pi},Z_{t-},H_{t-}))%
\bar{N}(dx,dt)\\
&=J(u,y,z,0)+\int_u^vJ_t(t,Y_t^{u,y,z,\pi},Z_t,H_t)dt+%
\int_u^vJ_y(t,Y_t^{u,y,z,\pi},Z_t,H_t)l(t)\sigma(Z_t)d W_{1t}\\
&+\int_u^vJ_y(t,Y_t^{u,y,z,\pi},Z_t,H_t)\bigg[r(Z_t)Y_t^{u,y,z,\pi}+(%
\mu(Z_t)-r(Z_t))l(t)+c^{(a)}\\
&+(1-H_t)m(t)\delta(1-\Delta)+m(t)\zeta(1-H_t)^2h^P\bigg]dt+%
\int_u^vJ_z(t,Y_t^{u,y,z,\pi},Z_t,H_t)g(Z_t)dt\\
&+\int_u^vJ_z(t,Y_t^{u,y,z,\pi},Z_t,H_t)\beta
d\tilde{W}_t+\frac{1}{2}\int_u^vJ_{yy}(t,Y_t^{u,y,z,\pi},Z_t,H_t)l^2(t)%
\sigma^2(Z_t)dt\\
&+\frac{1}{2}\int_u^vJ_{zz}(t,Y_t^{u,y,z,\pi},Z_t,H_t)\beta^2dt+%
\int_u^vJ_{yz}(t,Y_t^{u,y,z,\pi},Z_t,H_t)\rho\beta
l(t)\sigma(Z_t)dt\\&+\int_u^v(J(t,Y_t^{u,y,z,\pi}-m(t)%
\zeta,Z_t,1-H_t)-J(t,Y_{t-}^{u,y,z,\pi},Z_{t-},H_{t-}))dH_t\\
&+\int_u^v\int_0^\infty
(J(t,Y_t^{u,y,z,\pi}-\min(x,a),Z_t,H_t)-J(t,Y_{t-}^{u,y,z,\pi},Z_{t-},H_{t-}))%
\bar{N}(dx,dt) \end{aligned}
\end{equation}
where $\bar{N}$ is the Poisson random measure on $\mathbb{R}%
_{+}\times[0,\infty[$ defined by $\bar{N}=\mathop{\sum}\limits_{n\geq1}%
\delta_{(X_{n}, T_{n})}$.

Compensating (\ref{Ito's formula}) by
\begin{equation}
\begin{aligned} \lambda \int_u^v\int_0^\infty
(J(t,Y_t^{u,y,z,\pi}-\min(x,a),Z_t,H_t)-J(t,Y_{t-}^{u,y,z,%
\pi},Z_{t-},H_{t-}))dF(x)dt \\
\int_u^v(J(t,Y_t^{u,y,z,\pi}-m(t)\zeta,Z_t,1-H_t)-J(t,Y_{t-}^{u,y,z,\pi},Z_{t-},H_{t-})(1-H_t)h^P)dt \end{aligned}
\end{equation}
we obtain the following:
\begin{equation}
\begin{aligned}\label{Ito's formula 1}
 &J(v,Y_v^{u,y,z,\pi},Z_v,H_v)\\
=&J(u,y,z,0)+\int_u^v \mathcal L^{\pi}J(t,Y_t^{u,y,z,\pi},Z_{t-},H_{t-})dt\\
+&\int_u^vJ_y(t,Y_t^{u,y,z,\pi},Z_t,H_t)l(t)\sigma(Z_t)d W_{1t}+\int_u^vJ_z(t,Y_t^{u,y,z,\pi},Z_t,H_t)\beta d\tilde{W}_t\\
+&\int_u^v(J(t,Y_t^{u,y,z,\pi}-m(t)\zeta,Z_t,1-H_t)-J(t,Y_{t-}^{s,y,z,%
\pi},Z_{t-},H_{t-}))dM_t^P\\
+&\int_u^v\int_0^\infty(J(t,Y_t^{u,y,z,\pi}-\min(x,a),Z_t,H_t)-J(t,Y_{t-}^{u,y,z,%
\pi},Z_{t-},H_{t-}))(\bar{N}(dx,dt)-\lambda dF(x)dt) \end{aligned}
\end{equation}
The assumption of (\ref{assumption 2}), imply that all the stochastic
integrals with respect to the Brownian motion are martingales. By assumption
(\ref{assumption 1}):
\begin{equation*}
\int_{u}^{v}\int_{0}^{\infty}(J(t,Y_{t}^{u,y,z,\pi}-%
\min(x,a),Z_{t},H_{t})-J(t,Y_{t-}^{u,y,z,\pi},Z_{t-},H_{t-}))(\bar{N}%
(dx,dt)-\lambda dF(x)dt)
\end{equation*}
is a martingale (see Ikeda and Watanabe, 1989, p. 63). By assumption (\ref%
{assumption 3}):
\begin{equation*}
\int_{u}^{v}
(J(t,Y_{t}^{\pi}-m(t)\zeta,Z_{t},1-H_{t})-J(t,Y_{t-}^{%
\pi},Z_{t-},H_{t-}))dM_{t}^{P}
\end{equation*}
is a martingale. Then, taking expectations in (\ref{Ito's formula 1})
yields:
\begin{equation*}
E[J(v,Y_{v}^{\pi},Z_{v},H_{v})]=J(u,y,z,0)+E\bigg[\int_{u}^{v} \mathcal{L}%
^{\pi}F(t,Y_{t-}^{\pi},Z_{t-},H_{t})dt\bigg]
\end{equation*}
Since $F$ satisfies the HJB equation (\ref{HJB-pre}), we obtain that
\begin{equation}  \label{Ito2}
E[J(v,Y_{v}^{\pi},Z_{v},H_{v})]\leq J(u,y,z,0),
\end{equation}
and letting $v=T$ in (\ref{Ito2}), we get that
\begin{equation*}
J(u,y,z,0)\geq V(u,y,z,0).
\end{equation*}
To justify the second part of the theorem, we repeat the above calculations
for the strategy given by $\pi^{*}(t,Z_{t-})$. Then we have
\begin{equation*}
J(u,y,z,0)=E[U(Y_{T}^{\pi^{*}})|(Y_{u}^{\pi^{*}},Z_{u},H_{u}))=(y,z,0)]\leq
V(u,y,z,0),
\end{equation*}
and with the first part of the proof we get that
\begin{equation*}
J(u,y,z,0)=E[U(Y_{T}^{\pi^{*}})|(Y_{u}^{\pi^{*}},Z_{u},H_{u}))=(y,z,0)]=
V(u,y,z,0).
\end{equation*}
\end{proof}

\subsection{Period after default}

We define the pre-default and post-default value function by
\begin{equation}
\begin{aligned} V(t,y,z,h)=\left\{\begin{aligned} &V(t,y,z,0), {\rm
if}~h=0~({\rm the~pre~default~case}),\\ &V(t,y,z,1), {\rm if}~h=1~({\rm
the~post~default~case}), \end{aligned}\right. \end{aligned}
\end{equation}
and calculate the post-default case first.

When $h=1$, the HJB equation (\ref{HJB}) transforms into a relatively simple
form
\begin{equation}  \label{HJB-post}
\begin{aligned}
0=&J_t(t,y,z,1)+\mathop{\sup}\limits_{\pi\in\AA}\bigg%
\{J_y(t,y,z,1)[r(z)y+l(t)(\mu(z)-r(z))+c^{(a)}]\\
&+J_z(t,y,z,1)g(z)+\frac{1}{2}J_{yy}(t,y,z,1)l(t)^2\sigma(z)^2+%
\frac{1}{2}J_{zz}(t,y,z,1)\beta^2+J_{yz}(t,y,z,1)\beta\rho\sigma(z)l(t)\\
&+\lambda(EJ(t,y-\min(X_1,a),z,1)-EJ(t,y,z,1))\bigg\}\\
=&J_t(t,y,z,1)+\mathop{\sup}\limits_{l\in\rr}\bigg\{J_y(t,y,z,1)%
\left[r(z)y+l(t)(\mu(z)-r(z))\right]\\
&+J_z(t,y,z,1)g(z)+\frac{1}{2}J_{yy}(t,y,z,1)l(t)^2\sigma^2(z)+%
\frac{1}{2}J_{zz}(t,y,z,1)\beta^2+J_{yz}(t,y,z,1)\beta\rho\sigma(z)l(t)\bigg%
\}\\
&+\mathop{\sup}\limits_{a\in\rr}\bigg\{c^{(a)}J_y(t,y,z,1)+\lambda(
EJ(t,y-\min(X_1,a),z,1)-EJ(t,y,z,1))\bigg\} \end{aligned}
\end{equation}
with terminal condition $J(T,y,z,1)=U(y)$.

In order to obtain a linear PDE, in this work we considered only the case
where the correlation coefficient is equal to zero $(\rho=0)$.\newline
In addition to Hypothesis \ref{coefficient}, we assume the following:

\begin{hypothesis}
\label{hypothesis 2}1. $r(z)=r$ is constant;\newline
2. $g$ is uniformly Lipschitz and bounded;\newline
3. $\frac{(\mu(z)-r)^{2}}{\sigma^{2}(z)}$ bounded with a bounded first
derivative.
\end{hypothesis}

Due to the form of the utility function, we conjecture the following
function as a solution to the HJB equation (\ref{HJB-post}):
\begin{equation}  \label{post-solution}
f(t,y,z)=J(t,y,z,1)=-\xi(t,z)\exp\left\{ -\alpha ye^{r(T-t)}\right\} .
\end{equation}
where $\xi(t,z)$ is defined below as a solution to a Cauchy problem. From (%
\ref{post-solution}), we have:

\begin{equation}
\begin{aligned} &f_t(t,y,z)=\left(-\xi_t-\alpha yr\xi
e^{r(T-t)}\right)\exp\left\{-\alpha ye^{r(T-t)}\right\},\\
&f_y(t,y,z)=\alpha\xi e^{r(T-t)} \exp\left\{-\alpha ye^{r(T-t)}\right\},\\
&f_{yy}(t,y,z)=-\alpha^2\xi e^{2r(T-t)} \exp\left\{-\alpha
ye^{r(T-t)}\right\},\\ &f_z(t,y,z)=-\xi_z \exp\left\{-\alpha
ye^{r(T-t)}\right\},\\ &f_{zz}(t,y,z)=-\xi_{zz} \exp\left\{-\alpha
ye^{r(T-t)}\right\}. \end{aligned}
\end{equation}

\begin{equation}
\begin{aligned} &E\left[f(t,y-\min(X_1,a),z)-f(t,y,z)\right]\\ &=-\xi\alpha
e^{r(T-t)}\exp\left\{-\alpha ye^{r(T-t)}\right\}\int_0^a\exp\left\{\alpha
xe^{r(T-t)}\right\}\overline{F}(x)dx \end{aligned}
\end{equation}
(\ref{HJB-post}) becomes:\newline
\begin{equation}  \label{pre-cauchy}
\begin{aligned} 0=&-\xi_t-\frac{1}{2}\beta^2\xi_{zz}-g(z)\xi_z\\
&+\mathop{\sup}\limits_{a\in\rr}\left\{c^{(a)}\alpha\xi e^{r(T-t)}
-\lambda\xi\alpha e^{r(T-t)}\int_0^a\exp\left\{\alpha
xe^{r(T-t)}\right\}\overline{F}(x)dx\right\}\\
&+\mathop{\sup}\limits_{l\in\rr}\left\{-\frac{1}{2}l^2\sigma^2(z)\alpha^2\xi
e^{2r(T-t)}+(\mu(z)-r)l\alpha\xi e^{r(T-t)}\right\}. \end{aligned}
\end{equation}
Then by the first-order maximization conditions we obtain the maximum
\begin{align}  \label{max}
& l^{*}(t,z)=\frac{(\mu(z)-r)}{\alpha\sigma^{2}(z)}e^{-r(T-t)},  \notag \\
& a^{*}(t)=\frac{e^{-r(T-t)}}{\alpha}\ln(1+\theta).
\end{align}
Now, we substitute $l^{*}$ and $a^{*}$ in (\ref{max}) into (\ref{pre-cauchy}%
) derive the following Cauchy problem:
\begin{equation}  \label{Cauchy problem}
\left\{ \begin{aligned} &0=\xi_t+\frac{1}{2}\beta^2\xi_{zz}+g(z)\xi_z
-\bigg(\left[(\eta-\theta)\lambda\mu_\infty
+(1+\theta)\lambda\int_0^{a^*}\overline{F}(x)dx\right]\alpha e^{r(T-t)}\\
&-\lambda\alpha e^{r(T-t)}\int_0^{a^*}\exp\left\{\alpha
xe^{r(T-t)}\right\}\overline{F}(x)dx
+\frac{(\mu(z)-r)^2}{2\sigma^2(z)}\bigg)\xi\\ &\xi(T,z)=1. \end{aligned}%
\right.
\end{equation}

\begin{thm}
(Existence and Uniqueness Theorem)\label{UAE} Assume that
\begin{align}
& \int_{0}^{\infty}\exp\left\{ 8\alpha xe^{rT}\right\} dF(x)<\infty ,
\label{integability 1} \\
& \int_{0}^{\infty}x\exp\left\{ 8\alpha xe^{rT}\right\} dF(x)<\infty ,
\label{integability 2}
\end{align}
Then the Cauchy problem given by (\ref{Cauchy problem}) has a unique
classical solution $\hat{\xi}$, which satisfies the following conditions:
\begin{align}
& |\hat{\xi}(t,z)|\leq C_{1}(1+|z|),  \label{H-continuous 1} \\
& |\hat{\xi}_{z}(t,z)|\leq C_{2}(1+|z|),  \label{H-continuous 2}
\end{align}
where $C_{1}$ and $C_{2}$ are constants.
\end{thm}

\begin{proof}
: In order to prove this theorem, first we verify that the Cauchy problem
given by (\ref{Cauchy problem}) satisfies the conditions of Theorem \ref{A.1}
(see Appendix).

\textbf{Step 1}. Since $\beta$ is constant, then it is Lipschitz continuous,
H$\ddot{o}$lder continuous, and the operator $\frac{1}{2}\beta^{2}\partial
^{2}_{zz}$ is uniformly elliptic. By Hypothesis \ref{coefficient}, we know
that $g(z)$ is bounded and uniformly Lipschitz continuous.\newline
Now we prove that
\begin{align*}
h(t,z):= & \underbrace{\left[ (\eta-\theta)\lambda\mu_{\infty}+(1+\theta
)\lambda\int_{0}^{a^{*}}\overline{F}(x)dx\right] \alpha e^{r(T-t)}}%
_{h_{1}(t)} \\
& -\underbrace{\lambda\alpha e^{r(T-t)}\int_{0}^{a^{*}}\exp\left\{ \alpha
xe^{r(T-t)}\right\} \overline{F}(x)dx}_{h_{2}(t)} +\underbrace{\frac {%
(\mu(z)-r)^{2}}{2\sigma^{2}(z)}}_{h_{3}(z)}
\end{align*}
is bounded and uniformly H$\ddot{o}$lder continuous in compact subsets of $%
\mathbb{R}\times[0,T]$. By Hypothesis \ref{coefficient}, it is easy to check
that the last term $h_{3}(z)$ is bounded. The first term $h_{1}(t)$ is
bounded by $(1+\eta)\lambda\mu_{\infty}\alpha e^{rT}$ . In order to prove $%
h_{2}(t)$ is bounded, we observe that
\begin{align*}
h_{2}(t) & =\left| \lambda\alpha e^{r(T-t)}\int_{0}^{a^{*}}\exp\left\{
\alpha xe^{r(T-t)}\right\} \overline{F}(x)dx\right| \leq\lambda\alpha
e^{rT}\left\{ \left| \int_{0}^{a^{*}}\exp\{\alpha xe^{r(T-t)}\}\overline{F}%
(x)dx\right| \right\} \\
& \leq\lambda\alpha e^{rT}\left\{ \left| \int_{0}^{D}\exp\left\{ \alpha
xe^{r(T-t)}\right\} dx\right| +\left| \int_{0}^{D}\exp\left\{ \alpha
xe^{r(T-t)}\right\} F(x)dx\right| \right\} \\
& \leq2\lambda\alpha e^{rT}\int_{0}^{D}\exp\left\{ \alpha
xe^{r(T-t)}\right\} dF(x)\leq2\lambda\alpha
e^{rT}\int_{0}^{\infty}\exp\left\{ \alpha xe^{r(T-t)}\right\} dF(x) \\
& \le\infty
\end{align*}
thus $h(t,z)$ is bounded. \newline

\textbf{Step 2}. Now we prove that $h(z,t)$ is uniformly H$\ddot{o}$lder
continuous in compact subsets of $\mathbb{R}\times[0,T]$. For $h_{1}(t)$,
use the mean value theorem to obtain that for all $(t,t_{0})\in[0,T]\times[%
0,T]$:
\begin{align*}
|h_{1}(t)-h_{1}(t_{0})|= & \alpha(\theta-\eta)\lambda\mu_{\infty}\left|
e^{r(T-t)}-e^{r(T-t_{0})}\right| \\
& +(1+\theta)\lambda\alpha\left| \int_{0}^{a^{*}(t)}\overline{F}%
(x)dxe^{r(T-t)} -\int_{0}^{a^{*}(t_{0})}\overline{F}(x)dxe^{r(T-t_{0})}%
\right| \\
& \leq\left[ \alpha(\theta-\eta)\lambda\mu_{\infty}e^{rT}+(1+\theta
)\lambda\alpha e^{rT}\right] \left| t-t_{0}\right| ,
\end{align*}
then $h_{1}(t)$ is uniformly H$\ddot{o}$lder continuous.\newline
For $h_{2}(t)$, the mean value theorem implies that there exists $t_{1}\in
[t_{0},t]$ such that:
\begin{align*}
|h_{2}(t)-h_{2}(t_{0})| & =|\lambda\alpha
e^{r(T-t)}\int_{0}^{a^{*}(t)}\exp\left\{ \alpha xe^{r(T-t)}\right\}
\overline{F}(x)dx \\
& -\lambda\alpha e^{r(T-t_{0})}\int_{0}^{a^{*}(t_{0})}\exp\left\{ \alpha
xe^{r(T-t_{0})}\right\} \overline{F}(x)dx| \\
& =\bigg|-\lambda\alpha r e^{r(T-t_{1})}\int_{0}^{a^{*}(t_{1})}\exp\left\{
\alpha xe^{r(T-t_{1})}\right\} \overline{F}(x)dx \\
& -\lambda\alpha e^{r(T-t_{1})}\left[ \alpha
r\int_{0}^{a^{*}(t_{1})}x\exp\left\{ \alpha xe^{r(T-t_{1})}\right\}
\overline{F}(x)dx\right] \\
& +\exp\left\{ \alpha a^{*}{t_{1}}e^{r(T-t_{1})}\overline{F}(a^{*}(t_{1}))%
\frac{da^{*}(t_{1})}{dt_{1}}\right\} |t-t_{0}|\bigg| \\
& \leq\{r|h_{2}(t_{1})|+2re^{2rT}\int_{0}^{\infty}\alpha x\exp\left\{ \alpha
xe^{rT}\right\} dF(x) \\
& +2(1+\theta)\frac{r}{\alpha}e^{rT}\ln(1+\theta)\}|t-t_{0}|<\infty,
\end{align*}
We get that $h_{2}(t)$ is uniformly Lipschitz continuous in $[0,T]$. By
Hypothesis \ref{coefficient}, $h^{\prime}_{3}(z)$ is bounded, then $h_{3}(z)$
is uniformly H$\ddot{o}$lder continuous, i.e., for all $(z,z_{0})\in \mathbb{%
R}^{2}$\newline
\begin{equation*}
|h_{3}(z)-h_{3}(z_{0})|\leq C|z-z_{0}|^{1/2}.
\end{equation*}
Then $h(t,z)$ is uniformly H$\ddot{o}$lder continuous in compact subsets of $%
\mathbb{R\times}[0,T]$. So the Cauchy problem (\ref{Cauchy problem}) has a
unique solution $\hat{\xi}(t,z)$ which satisfies (\ref{H-continuous 1}) and (%
\ref{H-continuous 2}).
\end{proof}

The next theorem relates the value function with the HJB equation (\ref%
{HJB-post}).

\begin{thm}
(Post-Default Strategy). \label{Post-Strategy} If (\ref{integability 1}), (%
\ref{integability 2}) are satisfied, then the value function (when $h=1$)
defined by (\ref{HJB-post}) has the form:
\begin{equation}
V(t,y,z,1)=-\hat{\xi}(t,z)\exp\left\{ -\alpha ye^{r(T-t)}\right\} ,
\end{equation}
where $\hat{\xi}(t,z)$ is the unique solution of (\ref{Cauchy problem}), and

\begin{equation}
\left\{ \begin{aligned}
&l^*(t,z)=\frac{\mu(z)-r}{\alpha\sigma^2(z)}e^{-r(T-t)},\\ &m^*(t)=0,\\
&a^*(t)=\frac{\ln(1+\theta)}{\alpha}e^{-r(T-t)}, \end{aligned}\right.
\end{equation}
is the optimal reinsurance-investment strategy.
\end{thm}

\begin{proof}
: We have already checked that
\begin{equation}
f(t,y,z)=-\hat{\xi}(t,z)\exp\left\{ -\alpha ye^{r(T-t)}\right\} ,
\end{equation}
solves the HJB equation (\ref{HJB-post}). To prove that $f(t,y,z)$ is the
true value function, we shall verify that assumptions (\ref{assumption 1})-(%
\ref{assumption 2}) of the Theorem \ref{Verification Theorem} are satisfied
by $f(t,y,z)$.

\textbf{Step 1}. We consider the case in which $r = 0$. Let $\pi\in \mathcal{%
A}$ be an admissible strategy, then:
\begin{align*}
& \int_{0}^{\infty}\mathbb{E}\left|
f(t,Y_{t}^{\pi}-\min(x,a),Z_{t})-f(t,Y_{t}^{\pi},Z_{t})\right| ^{2}dF(x) \\
& =\int_{0}^{a}\mathbb{E}\left| -\hat{\xi} e^{-\alpha(Y_{t}^{\pi}-x})+\hat {%
\xi} e^{-\alpha Y_{t}^{\pi}}\right| ^{2}dF(x)+\int_{a}^{\infty}\mathbb{E}%
\left| -\hat{\xi} e^{-\alpha(Y_{t}^{\pi}-a})+\hat{\xi} e^{-\alpha
Y_{t}^{\pi}}\right| ^{2}dF(x) \\
& =\int_{0}^{a}\left( e^{\alpha x}-1\right) ^{2}dF(x)\mathbb{E}\left[ \hat {%
\xi}^{2}(t,Z_{t})\exp\left\{ -2\alpha Y_{t}^{\pi}\right\} \right]
+\int_{a}^{\infty}\left( e^{\alpha a}-1\right) ^{2}dF(x)\mathbb{E}[\hat{\xi}%
^{2}(t,Z_{t})\exp\{-2\alpha Y_{t}^{\pi} \}].
\end{align*}
To get condition (\ref{assumption 1}), we need only obtain an estimate of:
\begin{equation*}
\mathbb{E}[\hat{\xi}^{2}(t,Z_{t})\exp\{-2\alpha Y_{t}^{\pi}\}].
\end{equation*}
We observe that
\begin{align*}
& \mathbb{E}[\hat{\xi}^{2}(t,Z_{t})\exp\{-2\alpha Y_{t}^{\pi}\}]\leq
C_{1}^{2}\mathbb{E}[(1+|Z_{t}|)^{2}\exp\{-2\alpha Y_{t}^{\pi}\}] \\
& \leq C_{1}^{2}\{\mathbb{E}[(1+|Z_{t}|)^{4}]\}^{1/2}\{\mathbb{E}%
[\exp\{-4\alpha Y_{t}^{\pi}\}]\}^{1/2},
\end{align*}
and by Theorem A.2 in Badaoui and Fern\'{a}ndez (2013) \cite{BF}
\begin{equation*}
\mathbb{E}\bigg(\mathop{\sup}\limits_{0\leq t\leq T}Z_{t}^{4}\bigg)\leq
C_{2}(1+|z|^{4}).
\end{equation*}

So we can get that
\begin{align*}
& \{\mathbb{E}[(1+|Z_{t}|)^{4}]\}^{1/2}\leq\{\mathbb{E}(\sqrt{%
2(1+|Z_{t}|)^{2}})^{4}\}^{1/2} \\
& \leq\{4\mathbb{E}[(1+|Z_{t}|)^{4}]\}^{1/2}\leq2(1+C_{3}(1+|z|^{4}))^{1/2},
\end{align*}

From (\ref{risk process})we have
\begin{align*}
\mathbb{E}[\exp(-4\alpha Y_{t})] & \leq\mathbb{E}\left[ \exp\{-4\alpha%
\int_{0}^{t}l(s)\sigma(Z_{s})dW_{1s}
+4\alpha\sum_{i=1}^{N_{t}}\min(X_{i},a)\}\right] \\
& =\mathbb{E}\left[ \exp\left\{ \frac{1}{2}L_{t}
+16\alpha^{2}\int_{0}^{t}l^{2}(s)\sigma^{2}(Z_{s})ds
+4\alpha\sum_{i=1}^{N_{t}}\min(X_{i},a)\right\} \right] \\
& \leq e^{16\alpha^{2}C_{4}}\mathbb{E}\left[ \exp\left\{ \frac{1}{2}L_{t}
+4\alpha\sum_{i=1}^{N_{t}}\min(X_{i},a)\right\} \right] \\
& \leq e^{16\alpha^{2}C_{4}}\left\{ \mathbb{E}\left[ \exp\left\{
L_{t}\right\} \right] \right\} ^{1/2} \left\{ \mathbb{E}\left[ \exp\left\{
8\alpha \sum_{i=1}^{N_{t}}\min(X_{i},a) \right\} \right] \right\} ^{1/2}.
\end{align*}

where$L_{t}=-8\alpha\int_{0}^{t}l(s)\sigma(Z_{s})dW_{1s}
-32\alpha^{2}\int_{0}^{t}l^{2}(s)\sigma^{2}(Z_{s})ds.$

Since $\exp\{L_{t}\}$ is a martingale, we obtain:
\begin{align*}
\mathbb{E}[\exp(-4\alpha Y_{t})] & \leq e^{16\alpha^{2}C_{4}}\{\mathbb{E}%
[\exp\{8\alpha\sum_{i=1}^{N_{t}}\min(X_{i},a)\}]\}^{1/2} \\
& \leq e^{16\alpha^{2}C_{4}}\exp\{\frac{\lambda t}{2}(e^{8a\alpha}
-8\alpha\int_{0}^{a}e^{8a\alpha}F(x)dx)\} \\
&<\infty,
\end{align*}
which proves (\ref{assumption 1}).

\textbf{Step 2}. In order to prove conditions (\ref{assumption 2}), we
observe that:\newline
\begin{align*}
\mathbb{E}|f_{y}(s,Y_{s},Z_{s})|^{2}\leq C_{5}^{2}\mathbb{E}%
[(1+|Z_{t}|)^{4}\exp\{-4\alpha Y_{s}^{\pi}\}]
\end{align*}
and
\begin{align*}
\mathbb{E}|f_{z}(t,Y_{t},Z_{t})|^{2}\leq C_{6}^{2}\mathbb{E}%
[(1+|Z_{t}|)^{4}\exp\{-4\alpha Y_{s}^{\pi}\}].
\end{align*}
Then by the same arguments as above, we get conditions (\ref{assumption 2})
and (\ref{assumption 3}). For the case in which the interest rate $r\neq0$,
let $\widetilde{Y}_{t}^{\pi}=e^{r(T-t)}Y_{t}^{\pi}$. An application of It$%
\hat{o}$'s formula shows that $\widetilde{Y}_{t}^{\pi}$ satisfies the
following SDE:\newline
\begin{align*}
d\widetilde{Y}_{t}^{\pi}= & e^{r(T-t)} \bigg[(\eta-\theta)\lambda\mu_{\infty
}+(1+\theta)\lambda\int_{0}^{a}\overline{F}(x)dx+Y(t)r(Z_{t})+(\mu(Z_{t}) \\
& -r(Z_{t}))l(t)\bigg]dt+e^{r(T-t)}l(t)\sigma(Z_{t})dW_{1t}-e^{r(T-t)}d%
\left( \sum_{i=1}^{N_{t}}\min(X_{i},a)\right) ,
\end{align*}
the result can be derived in a similar way as in the first part of the proof.
\end{proof}

\subsection{Period before default}

In this subsection, we will focus on the pre-default case. When $h=0$, the
HJB equation (\ref{HJB}) transforms into
\begin{equation}  \label{HJB-pre}
\begin{aligned}
0=&J_t(t,y,z,0)+\mathop{\sup}\limits_{\pi\in\AA}\bigg\{\left[r(z)y+l(t)(%
\mu(z)-r(z))+c^{(a)}+m(t)\delta\right]J_y(t,y,z,0)\\
&+J_z(t,y,z,0)g(z)+\frac{1}{2}J_{yy}(t,y,z,0)l(t)^2\sigma(z)^2+%
\frac{1}{2}J_{zz}(t,y,z,0)\beta^2+J_{yz}(t,y,z,0)\beta\rho\sigma(z)l(t)\\
&+\lambda \left(EJ(t,y-\min(X_1,a),z,0)-EJ(t,y,z,0)\right)\\
&+\left(J(t,y-m(t)\zeta,z,1)-J(t,y,z,0)\right)h^P\bigg\} \end{aligned}
\end{equation}
with terminal condition $J(T,y,z,0)=U(y)$.

According to Fleming and Soner (1993), if the optimal value function $%
V(t,y,z,0)\in C^{1,2,2}([0,T]\times\mathbb{R}\times\mathbb{R})$, then $V$
satisfies the HJB equation (\ref{HJB-pre}). To solve this equation, take as
a trial solution
\begin{equation}
\bar{f}(t,y,z)=J(t,y,z,0)=-\bar{\xi}(t,z)\exp\{-\alpha ye^{r(T-t)}\},
\end{equation}
with $\bar{\xi}(T,z)=1$. Then we have:

\begin{equation}  \label{pre-F}
\begin{aligned} &\bar{f}_t(t,y,z)=(-\bar{\xi}_t-\alpha yr\bar{\xi}
e^{r(T-t)})\exp\{-\alpha ye^{r(T-t)}\},\\ &\bar{f}_y(t,y,z)=\alpha\bar{\xi}
e^{r(T-t)} \exp\{-\alpha ye^{r(T-t)}\},\\
&\bar{f}_{yy}(t,y,z)=-\alpha^2\bar{\xi} e^{2r(T-t)} \exp\{-\alpha
ye^{r(T-t)}\},\\ &\bar{f}_z(t,y,z)=-\bar{\xi}_z \exp\{-\alpha
ye^{r(T-t)}\},\\ &\bar{f}_{zz}(t,y,z)=-\bar{\xi}_{zz} \exp\{-\alpha
ye^{r(T-t)}\}. \end{aligned}
\end{equation}
and
\begin{equation}  \label{E_F}
\begin{aligned} &E[\bar{f}(t,y-\min(X_1,a),z)-\bar{f}(t,y,z)]\\
&=-\bar{\xi}\alpha e^{r(T-t)}\exp\{-\alpha
ye^{r(T-t)}\}\int_0^a\exp\left\{\alpha xe^{r(T-t)}\right\}\overline{F}(x)dx,
\end{aligned}
\end{equation}

\begin{equation}  \label{default_F}
\begin{aligned} &({f}(t,y-m(t)\zeta,z)-\bar{f}(t,y,z))h^P\\
&=-\hat{\xi}(z,t)\exp\{-\alpha(y-m(t)\zeta)e^{r(T-t)}\}h^P
+\bar{\xi}(z,t)\exp\{-\alpha ye^{r(T-t)}\}h^P, \end{aligned}
\end{equation} where $\hat{\xi}$ is the unique classical solution of the Cauchy problem (%
\ref{Cauchy problem}).
Substituting the above formulas (\ref{pre-F})-(\ref{default_F}) into (\ref%
{HJB-pre}), when $\rho=0$, we have
\begin{equation}  \label{PDE1}
\begin{aligned}
0=&-\bar{\xi}_t-\frac{1}{2}\beta^2\bar{\xi}_{zz}-g(z)\bar{\xi}_z
+(\eta-\theta)\lambda\mu_\infty\alpha\bar{\xi}e^{r(T-t)}\\
&+\mathop{\sup}\limits_{l}\bigg\{(\mu(z)-r)\alpha \bar{\xi}
e^{r(T-t)}l-\frac{1}{2}\alpha^2\bar{\xi}e^{2r(T-t)}\sigma^2(z)l^2\bigg\}\\
&+\mathop{\sup}\limits_{m}\bigg\{m\delta\alpha\bar{\xi}
e^{r(T-t)}+(\bar{\xi}-e^{\alpha m\zeta e^{r(T-t)}}\hat{\xi})h^P\bigg\}\\
&+\mathop{\sup}\limits_{a}\bigg\{(1+\theta)\lambda\int_0^a\overline{F}(x)dx%
\alpha\overline{\xi}
e^{r(T-t)}-\lambda\alpha\bar{\xi}e^{r(T-t)}\int_0^a\exp\left\{\alpha
xe^{r(T-t)}\right\}\overline{F}(x)dx\bigg\}. \end{aligned}
\end{equation}
According to Theorem \ref{Post-Strategy}, the first-order conditions for a
regular interior maximization in (\ref{PDE1}) are

\begin{equation}  \label{pre-optimal}
\left\{ \begin{aligned}
&l^*(t,z)=\frac{\mu(z)-r}{\alpha\sigma^2(z)}e^{-r(T-t)},\\
&m^*(t,z)=\frac{\ln\frac{1}{\triangle}+\ln\overline{\xi}-\ln\hat{\xi}}{%
\alpha \zeta}e^{-r(T-t)},\\ &a^*(t)=\frac{\ln(1+\theta)}{\alpha}e^{-r(T-t)}.
\end{aligned}\right.
\end{equation}
where $\hat{\xi}$ is the unique classical solution of the Cauchy problem (%
\ref{Cauchy problem}).

We now insert (\ref{pre-optimal}) into (\ref{PDE1}), thereby obtaining
\begin{equation}  \label{PDE2}
\begin{aligned}
0=&\overline{\xi}_t+\frac{1}{2}\beta^2\overline{\xi}_{zz}+g(z)\overline{%
\xi}_z -\frac{h^P}{\Delta}\bar{\xi}\ln\bar{\xi}\\
&-\bigg\{\left[(\eta-\theta)\mu_\infty+(1+\theta)\int_0^{a^*(t)}%
\overline{F}(x)dx\right]\lambda\alpha e^{r(T-t)}\\ &-\lambda\alpha
e^{r(T-t)}\int_0^{a^*(t)}\exp\{\alpha xe^{r(T-t)}\}\overline{F}(x)dx\\
&+\frac{(\mu(z)-r)^2}{2\sigma^2(z)}
+\left(1-\frac{1}{\Delta}+\frac{1}{\Delta}\ln\frac{1}{\Delta}\right)h^P
-\frac{h^P}{\Delta}\ln\hat{\xi}\bigg\}\bar{\xi}, \end{aligned}
\end{equation}

We let
\begin{align*}
M(t,z) & =\bigg[(\eta-\theta)\mu_{\infty}+(1+\theta)\int_{0}^{a^{*}(t)}%
\overline{F}(x)dx\bigg]\lambda\alpha e^{r(T-t)} \\
& -\lambda\alpha e^{r(T-t)}\int_{0}^{a^{*}(t)}\exp\{\alpha xe^{r(T-t)}\}%
\overline{F}(x)dx \\
& +\frac{(\mu(z)-r)^{2}}{2\sigma^{2}(z)} +\underbrace{(1-\frac{1}{\Delta }+%
\frac{1}{\Delta}\ln\frac{1}{\Delta})h^{P}}_{I} -\frac{h^{P}}{\Delta }%
\underbrace{\ln\hat{\xi}}_{\hat{u}}, \\
& =h(t,z)+I-\frac{h^{P}}{\Delta}\hat{u}.
\end{align*}
where $h(t,z)$ is defined in the proof of Theorem \ref{UAE}. Then, according
to hypothesis \ref{coefficient}, $M(t,z)$ is bounded and (\ref{PDE2})
becomes
\begin{equation}  \label{PDE3}
\begin{aligned}
0=&\overline{\xi}_t+\frac{1}{2}\beta^2\overline{\xi}_{zz}+g(z)\overline{%
\xi}_z -\frac{h^P}{\Delta}\bar{\xi}\ln\bar{\xi}-M(t,z)\bar{\xi}.
\end{aligned}
\end{equation}

In order to solve this PDE, we make variable substitution $\bar{u}=\ln\bar {%
\xi}$, then $\bar{u}(T,z)=0$ and we have
\begin{equation}  \label{new varible}
\begin{aligned} &\bar{\xi}=e^{\bar{u}},\\
&\bar{\xi}_t={\bar{u}}_te^{\bar{u}},\\
&\bar{\xi}_z={\bar{u}}_ze^{\bar{u}},\\
&\bar{\xi}_{zz}=({\bar{u}}_z^2+{\bar{u}}_{zz})e^{\bar{u}}, \end{aligned}
\end{equation}

Substituting the above formulas (\ref{new varible}) into (\ref{PDE3}), we
get
\begin{equation}  \label{PDE4}
\left\{ \begin{aligned}
&0=\bar{u}_t+\frac{1}{2}\beta^2(\bar{u}_{zz}+\bar{u}_z^2)+g(z)\bar{u}_z-%
\frac{h^P}{\Delta}\bar{u}-M(t,z).\\ &\bar{u}(T,z)=0 \end{aligned}\right.
\end{equation}
Eq. (\ref{PDE4}) is indeed a Cauchy initial value problem (CIVP).

Use the same transform
\begin{equation}  \label{U-solution}
u=\ln\xi
\end{equation}
we rewrite CIVP (\ref{Cauchy problem}) as
\begin{equation}  \label{PDE00}
\left\{ \begin{aligned}
&0=u_t+\frac{1}{2}\beta^2(u_{zz}+u_z^2)+g(z)u_z-h(t,z).\\ &u(T,z)=0.
\end{aligned}\right.
\end{equation}

In order to solve CIVP (\ref{PDE4}), we found that technical complications
in quasi-linear parabolic PDEs (\ref{PDE4}) are generated by the quadratic
growth of the gradient. Due to the nonlinearity of (\ref{PDE4}), we consider
the so-called super-sub solution method as in Birge, Bo and Capponi(2016),
see Bebernes and Schmitt(1977) and Bebernes and Schmitt (1979) for the
general theory in the parabolic case, and establish the so-called ordered
pair of lower and upper solutions to the CIVP (\ref{PDE4}). The definition
of lower and upper solutions to the CIVP (\ref{PDE4}) is given as follows
(see also Bebernes and Schmitt (1979) and Birge, Bo and Capponi(2016)).

Let
\begin{equation}
\begin{aligned}
&L\upsilon(t,z)=\upsilon_t+\frac{1}{2}\beta^2\upsilon_{zz}+g(z)\upsilon_z-%
\frac{h^P}{\Delta}\upsilon\\
&G(t,z,\upsilon,p)=-\frac{1}{2}\beta^2p^2+M(t,z)\end{aligned}
\end{equation}

\begin{defn}
A continuous function $\varphi:(0,T)\times\mathbb{R}\rightarrow\mathbb{R}$
is called a lower solution of the CIVP(\ref{PDE4}) if $\varphi(T,z)\leq0$
for $z\in\mathbb{R}$, and for every $(z_{0},t_{0}) \in(0,T)\times\mathbb{R}$
there exists an open neighborhood $\mathcal{O}$ of $(z_{0},t_{0})$ such that
for $(t,z)\in\mathcal{O}\cap(0,T)\times\mathbb{R}$,
\begin{equation}
\begin{aligned} L\varphi \geq G(t,z,\varphi,\varphi_z). \end{aligned}
\end{equation}
If in the above expression the inequality sign is reversed, then $\varphi$
is called an upper solution of the CIVP (\ref{PDE4}). Let $\bar{\varphi}$
and $\underline{\varphi}$ be the upper and lower solution respectively. If $%
\underline{\varphi}(t,z)\leq\bar{\varphi}(t,z)$ for all $(t,z)\in [0,T]\times%
\mathbb{R}$, we call $(\underline{\varphi},\bar{\varphi})$ an ordered pair
of lower and upper solutions of the CIVP (\ref{PDE4}).
\end{defn}

We next construct lower and upper solutions to the CIVP (\ref{PDE4}). In
Theorem \ref{UAE}, we have already proven that $\hat{\xi}$ is the
nonnegative classical solution of the CIVP (\ref{Cauchy problem}), so $\hat{u%
}$ is the classical solution of the CIVP (\ref{PDE00}). Let

\begin{equation}  \label{upper solution}
\bar{\varphi}(t,z)=\hat{u}(t,z),
\end{equation}
we have

\begin{equation}
\begin{aligned}
&L\bar{\varphi}=\bar{\varphi}_t+\frac{1}{2}\beta^2\bar{\varphi}_{zz}++g(z)%
\bar{\varphi}_z-\frac{h^P}{\Delta}\bar{\varphi}
=-\frac{1}{2}\beta^2{\bar{\varphi}}_z^2+M(t,z)-I\\
&G(t,z,\bar{\varphi},\bar{\varphi}_z)=-\frac{1}{2}\beta^2\bar{%
\varphi}_z^2+M(t,z) \end{aligned}
\end{equation}
Since $1-x\leq e^{-x}$ for any real number we get that $I\geq0$, so $\bar{%
\varphi}$ is an upper solution of the CIVP (\ref{PDE4}).

Let
\begin{equation}  \label{lower solution}
\underline{\varphi}(t,z)=\hat{u}(t,z)-\frac{\Delta }{h^{P}}I,
\end{equation}
so
\begin{equation}
\begin{aligned}
&L\underline{\varphi}=-\frac{1}{2}\beta^2\hat{u}_z^2+M(t,z)\\
&G(t,z,\underline{\varphi},{\underline{\varphi}}_z)=-\frac{1}{2}\beta^2%
\hat{u}_z^2+M(t,z) \end{aligned}
\end{equation}
then we have $L\underline{\varphi}= G(t,z,\underline{\varphi},{\underline{%
\varphi}}_{z})$ and $\underline{\varphi}(T,z)\le0$, it follows that $%
\underline{\varphi}$ is a lower solution to the CIVP (\ref{PDE4}). Moreover,
$(\underline{\varphi},\bar{\varphi})$ is an ordered pair of lower and upper
solution of the CIVP (\ref{PDE4}). We are now ready to give the main result
of the paper, which establishes the existence of classical solutions to the
CIVP (\ref{PDE4}).

\begin{thm}
\label{Existence Theorem}(Existence Theorem)  If (\ref{integability 1}), (%
\ref{integability 2}) and Hypothesis (\ref%
{coefficient}-\ref{hypothesis 2}) are satisfied. Then there exists a
classical solution $\tilde{u}$ to CIVP(\ref{PDE4}). Moreover, it holds that
\begin{equation}  \label{lower-up bound}
\underline{\varphi}(t,z)\leq\tilde{u}(t,z)\leq \bar{\varphi}(t,z)
\end{equation}
where $\bar{\varphi}$ and $\underline{\varphi}$ are defined in (\ref{upper
solution}) and (\ref{lower solution}), respectively. Additionally the Cauchy
problem given by (\ref{PDE3}) exists a classical solution $\tilde{\xi}$,
which satisfies the following conditions:
\begin{align}
& |\tilde{\xi}(t,z)|\leq C_{7}(1+|z|),  \label{H-continuous 3} \\
& |\tilde{\xi}_{z}(t,z)|\leq C_{8}(1+|z|),  \label{H-continuous 4}
\end{align}
where $C_{7}$ and $C_{8}$ are constants.
\end{thm}

\begin{proof}
We follow the proof in Theorem 4.2 of Birge, Bo and Capponi(2016). From the
above analysis we know that $(\underline{\varphi},\bar{\varphi})$ is an
ordered pair of lower and upper solution of the CIVP (\ref{PDE4}). Next, if $%
\tilde{u}$ is the classical solution to the CIVP (\ref{PDE4}), using an
invariance result (see, e.g. Lemma 1 of Bebernes and Schmitt (1979)), it
follows that $\tilde{u}(t,z)\in[\underline{\varphi}(t,z),\bar{\varphi}(t,z)]$
for all $(t,z)\in[0,T]\times\mathbb{R}$. Let $R>0$ be an arbitrary constant
and $B_{R}:=\{q:\in\mathbb{R}; |q|<R\}$. Therefore, for all $\upsilon \in[%
\underline{\varphi}(t,z),\bar{\varphi}(t,z)]$ and $(t,z)\in[0,T]\times \bar{B%
}_{R}$, we obtain that
\begin{equation}
\begin{aligned}\label{G bound} |G(t,z,\upsilon,p)|&\leq \frac{1}{2}
\beta^2p^2+|h(t,z)|+H+\frac{h^p}{\Delta}|\hat{u}(t,z)|\\ &\leq K_R(1+|p|^2)
\end{aligned}
\end{equation}
where $K_{R}>0$ is a generic constant which depends on $R$. This shows that
the cofficient $f$ admits the quadratic growth in $p$. However, $f$ fails to
satisfy a Nagumo type condition. (See Theorem 2 of Bebernes and Schmitt
(1979) where this condition is treated and it is required that $%
|f(t,y,\upsilon ,p)|\leq\Phi(|p|)$ for some positive continuous
nondecreasing function $\Phi$ such that $\mathop{\lim}\limits_{s\rightarrow%
\infty}\frac{s^{2}}{\Phi (s)}=\infty$. In our case $\Phi(s)=s^{2}$ does not
admit, given that $\mathop{\lim}\limits_{s\rightarrow\infty}\frac{s^{2}}{%
\Phi(s)}=1$.) Hence, Theorem 3 of Bebernes and Schmitt (1979) is not
applicable for our case. To overcome this, we adopt an approximation
technique used in Loc and Schmitt(2012) which extends the Nagumo conditions
to Bernstein-Nagumo conditions. The latter covers the quadratic growth
condition of $G$ in $p$ given in Eq. (\ref{G bound}). As in Loc and Schmitt
(2012), for $k\in \mathbb{N}$, we define a truncated function $h_{k}(p)$
acting on $p\in\mathbb{R}$ as

\begin{equation}
\begin{aligned} h_k(p)=\left\{\begin{aligned} &p, {\rm if} ~~~|p|\leq k,\\
&\frac{k}{|p|}, \end{aligned}\right. \end{aligned}
\end{equation}
Then we consider the following PDE given by
\begin{equation}  \label{truncate eq}
\begin{aligned}
(u_k)_t+\frac{1}{2}\beta^2(u_k)_{zz}+g(z)(u_k)_z-\frac{H^P}{%
\Delta}u_k-G_k(t,z,u_k,(u_k)_z)=0 \end{aligned}
\end{equation}
where $G_{k}(t,z,\upsilon,p):=-\frac{1}{2}\beta^{2}h_{k}(p)^{2}+M(t,z)$. It
can be easily seen that, for each $k\in\mathbb{N}$ and $R>0$, $%
G_{k}(t,z,\upsilon,p)$ satisfies the Nagumo growth condition in $p$ required
by theorem 3 of Bebernes and Schmitt(1979), for all $\upsilon\in [\underline{%
\varphi}(t,z),\bar{\varphi}(t,z)]$ with $(t,z)\in[0,T]\times \bar{B}_{R}$.
Then we can apply theorem 3 of Bebernes and Schmitt(1979), and deduce that
Eq. \ref{truncate eq} admits a solution $\tilde{u}_{k}(t,z)$, $(t,z)\in[0,T]%
\times\mathbb{R}$, in the classic sense for each $k\in \mathbb{N}$. Notice
that $G_{k}(t,z,\upsilon,p)\rightarrow G(t,z,\upsilon,p)$ pointwise as $%
k\rightarrow\infty$. Then we can extract a subsequence of $\tilde{u}%
_{k_{l}}(t,z)$ which converges uniformly on compact subsets of $[0,T]\times%
\mathbb{R}$ to a solution of the CIVP (\ref{HJB-pre}). Moreover the limit of
the above subsequence of $\tilde{u}_{k_{l}}(t,z)$ also lies in $[\underline{%
\varphi}(t,z),\bar{\varphi}(t,z)]$ for all $(t,z)\in [0,T]\times\mathbb{R}$.
We write the limit is $\tilde{u}(t,z)$ and $\tilde {\xi}(t,z)=e^{\tilde{u}%
(t,z)}$. From (\ref{lower-up bound}), we know that
\begin{equation}
e^{\frac{\Delta}{h^{P}}I}\hat{\xi}=e^{\underline{\varphi}(t,z)}\leq\tilde{%
\xi }(t,z)\leq e^{\bar{\varphi}(t,z)}=\hat{\xi}
\end{equation}
This completes the proof of the theorem.
\end{proof}

\begin{thm}
(Pre-Default Strategy). \label{Pre-Strategy} If (\ref{integability 1}), (\ref%
{integability 2}) are satisfied, then the value function (when $h=0$)
defined by (\ref{HJB-pre}) has the form:
\begin{equation}
V(t,y,z,0)=-\tilde{\xi}(t,z)\exp\left\{ -\alpha ye^{r(T-t)}\right\} ,
\end{equation}
The optimal investment strategy is given by $\tilde{\pi}_{t}^{*}=\pi
^{*}(t,Z_{t}-)$, where the optimal feedback control function is given as
follows:
\begin{equation}
\left\{ \begin{aligned}
&l^*(t,z)=\frac{\mu(z)-r}{\alpha\sigma^2(z)}e^{-r(T-t)},\\
&m^*(t,z)=\frac{\ln\tilde{\xi}(t,z)-\ln\hat{\xi}(t,z)+\ln\frac{1}{\Delta}}{%
\alpha\zeta}e^{-r(T-t)},\\ &a^*(t)=\frac{\ln(1+\theta)}{\alpha}e^{-r(T-t)}.
\end{aligned}\right.
\end{equation}
where $\hat{\xi}$ is the unique solution of CIVP (\ref{Cauchy problem}) and $%
\tilde{\xi}$ is the unique solution of DPE (\ref{PDE3}) with terminal condition $%
\tilde{\xi}(T,z)=1$.
\end{thm}

\begin{proof}
: The proof is the same as the post-default case. We have already checked
that
\begin{equation}
J(t,y,z,0)=\bar{f}(t,y,z)=-\tilde{\xi}(t,z)\exp\left\{ -\alpha
ye^{r(T-t)}\right\} ,
\end{equation}
solves the HJB equation (\ref{HJB-post}). To prove that $\bar{f}(t,y,z)$ is
the true value function, we shall verify that assumptions (\ref{assumption 1}%
)-(\ref{assumption 3}) of the Theorem \ref{Verification Theorem} are
satisfied by $\bar{f}(t,y,z)$.

\textbf{Step 1}. We consider the case in which $r = 0$. Let $\pi\in \mathcal{%
A}$ be an admissible strategy, then:
\begin{align*}
& \int_{0}^{\infty}\mathbb{E}\left| \bar{f}(t,Y_{t}^{\pi}-\min(x,a),Z_{t})-%
\bar{f}(t,Y_{t}^{\pi},Z_{t})\right| ^{2}dF(x) \\
& =\int_{0}^{a}\mathbb{E}\left| -\tilde{\xi} e^{-\alpha(Y_{t}^{\pi}-x})+%
\tilde{\xi} e^{-\alpha Y_{t}^{\pi}}\right| ^{2}dF(x)+\int_{a}^{\infty }%
\mathbb{E}\left| -\tilde{\xi} e^{-\alpha(Y_{t}^{\pi}-a})+\tilde{\xi}
e^{-\alpha Y_{t}^{\pi}}\right| ^{2}dF(x) \\
& =\int_{0}^{a}\left( e^{\alpha x}-1\right) ^{2}dF(x)\mathbb{E}\left[ \tilde{%
\xi}^{2}(t,Z_{t})\exp\left\{ -2\alpha Y_{t}^{\pi}\right\} \right]
+\int_{a}^{\infty}\left( e^{\alpha a}-1\right) ^{2}dF(x)\mathbb{E}[\tilde{%
\xi }^{2}(t,Z_{t})\exp\{-2\alpha Y_{t}^{\pi} \}].
\end{align*}
To get condition (\ref{assumption 1}), we need only obtain an estimate of:
\begin{equation*}
\mathbb{E}[\tilde{\xi}^{2}(t,Z_{t})\exp\{-2\alpha Y_{t}^{\pi}\}].
\end{equation*}
From (\ref{risk process})we have
\begin{equation*}
\mathbb{E}[\exp(-4\alpha Y_{t}^{\pi})]\leq\mathbb{E}\left[ \exp\left\{
-4\alpha\int_{0}^{t}l(s)\sigma(Z_{s})dW_{1s} +4\alpha\int_{0}^{t}
m(s)(1-H_{s})\zeta dM_{s}^{P}+4\alpha\sum_{i=1}^{N_{t}}\min(X_{i},a)\right\} %
\right]
\end{equation*}
By Step 1 in theorem (\ref{Post-Strategy}), we only need to estimate
\begin{equation*}
\mathbb{E}\exp\left\{ 4\alpha\int_{0}^{t} m(s)(1-H_{s})\zeta
dM_{s}^{P}\right\} .
\end{equation*}
because of that
\begin{align*}
& \mathbb{E}\bigg[\exp\bigg(\int_{0}^{t} m(s)(1-H_{s})\zeta dM_{s}^{P}\bigg)%
\bigg] \\
& \leq\mathbb{E}\bigg[\exp\bigg(\int_{0}^{t} m(s)(1-H_{s})\zeta dH_{s}\bigg)%
\bigg] \\
& \leq\mathbb{E}\bigg[\exp\bigg(\int_{0}^{t} m(s)\zeta dH_{s}\bigg)\bigg] \\
& \leq\exp^{\int_{0}^{t}(e^{m(s)\zeta}-1)h^{P}ds}.
\end{align*}
From (\ref{lower-up bound}) in theorem \ref{Existence Theorem}, we know that
\begin{equation*}
1-\Delta-\ln\frac{1}{\Delta}=-\frac{\Delta}{h^{P}}I\leq\tilde{u}-\hat{u}\leq0
\end{equation*}
Then we have the lower and upper bound of $m^{*}(t,z)$ is that
\begin{align*}
m^{*}(t,z) & =\frac{\ln\tilde{\xi}(t,z)-\ln\hat{\xi}(t,z)+\ln\frac{1}{\Delta}%
}{\alpha\zeta}e^{-r(T-t)} \\
& =\frac{\tilde{u}-\hat{u}+\ln\frac{1}{\Delta}}{\alpha\zeta}e^{-r(T-t)} \\
0\leq\frac{1-\Delta}{\alpha\zeta}e^{-r(T-t)} & \leq m^{*}(t,z)\leq\frac {\ln%
\frac{1}{\Delta}}{\alpha\zeta}e^{-r(T-t)}
\end{align*}

which proves (\ref{assumption 1}).

\textbf{Step 2}. It is the same as Step 2 in theorem \ref{Post-Strategy}
which proves (\ref{assumption 2}).

\textbf{Step 3}. By Lemma \ref{uniformly integrability}, we know that
\begin{equation*}
J({\tau_{i}\wedge T}, Y^{\pi^{*}}_{\tau_{i}\wedge T}-m^{*}(\tau_{i}\wedge
T)\zeta, Z_{\tau_{i}\wedge T}, 1-H_{\tau_{i}\wedge T})-J({\tau_{i}\wedge T},
Y^{\pi^{*}}_{\tau_{i}\wedge T}, Z_{\tau_{i}\wedge T}, H_{\tau_{i}\wedge T})
\end{equation*}
is uniformly integrable which proves (\ref{assumption 3}).
\end{proof}

\begin{lem}
\label{uniformly integrability} Let $\tau_{i}$ be the exist time of $%
(Y_{t},Z_{t},H_{t})$ from the open set $M_{i}$ , where $M_{i}\subset M = [ 0
,\infty)\times[ 0 ,\infty)\times\{0,1\}$ such that $M_{i}\subset M_{i +
1}\subset M$, $i\in N^{+}$ , and $M =\cup_{i} M_{i}$ . Then we have
\begin{equation}
\begin{aligned} &\mathop{\sup}\limits_{i}E\bigg[|J({\tau_i\wedge T},
Y^{\pi^*}_{\tau_i\wedge T}-m(\tau_i\wedge T)\zeta, Z_{\tau_i\wedge T},
1-H_{\tau_i\wedge T})|^2\bigg]< \infty, i\in N^+ . \\
&\mathop{\sup}\limits_{i}E\bigg[|J({\tau_i\wedge T}, Y^{\pi^*}_{\tau_i\wedge
T}, Z_{\tau_i\wedge T}, H_{\tau_i\wedge T})|^2\bigg] < \infty, i\in N^+ .
\end{aligned}
\end{equation}
i.e.%
\begin{equation*}
J({\tau_{i}\wedge T}, Y^{\pi^{*}}_{\tau_{i}\wedge T}-m^{*}(\tau_{i}\wedge
T)\zeta, Z_{\tau_{i}\wedge T}, 1-H_{\tau_{i}\wedge T})-J({\tau_{i}\wedge T},
Y^{\pi^{*}}_{\tau_{i}\wedge T}, Z_{\tau_{i}\wedge T}, H_{\tau_{i}\wedge T})
\end{equation*}
is uniformly integrable.
\end{lem}

\begin{proof}
: In view of Eq.(\ref{risk process}), the wealth process associated with the
strategy $\pi^{*}$ is

\beq\label{risk process 1}
\begin{aligned}
Y^{\pi^*}_{t}&=y+\int_0^t[r(Z_{t})Y^{\pi^*}_{t}+(\mu(Z_{t})-r(Z_{t}))l(t)+c^{(a)}+(1-H_{t})m(t)\delta]dt\\
&+\int_0^t l(t)\sigma(Z_{t})dW_{1t}-\int_0^t m(t)(1-H_{t})\zeta dM^{P}_{t}-\sum_{i=1}^{N_{t}}\min(X_{i},a(t)). \end{aligned}\nneq

Let
$$\bar{Y}^{*}_{t}=e^{-rt}Y^{\pi^*}_{t}.$$

An application of It$\hat{o}$'s formula leads to

\beq\label{risk process 2}
\begin{aligned}
\bar{Y}^{*}_{t}&=y+\int_0^te^{-rs}dY^{\pi^*}_{s}+\int_0^t(-r)e^{-rs}Y^{\pi^*}_{s}ds\\
&=y+\int_0^t[e^{-rs}(\mu(Z_{s})-r))l^*(s)+c^{(a^*(s))}+(1-H_{s})m^*(s)\delta(1-\Delta)]ds\\
&+\int_0^t e^{-rs}l^*(s)\sigma(Z_{s})dW_{1s}-\int_0^t e^{-rs}m^*(s)(1-H_{s})\zeta dM^{P}_{s}-\int_0^t e^{-rs}d\sum_{i=1}^{N_{s}}\min(X_{i},a^*(s))\\ &=y+\int_0^te^{-rT}\bigg[\frac{(\mu(Z_{s})-r))^2}{\alpha\sigma^2(Z_s)}+c^{(a^*(s))}+(1-H_{s})
\frac{\ln\tilde{\xi}(s,Z_s)-\ln\hat{\xi}(s,Z_s)+\ln\frac{1}{\Delta}}{\alpha\zeta}\delta\bigg]ds\\
&+\int_0^t e^{-rT}\frac{\mu(Z_{s})-r}{\alpha\sigma(Z_s)}dW_{1s}-\int_0^t e^{-rT}\frac{\ln\tilde{\xi}(s,Z_s)-\ln\hat{\xi}(s,Z_s)+\ln\frac{1}{\Delta}}{\alpha\zeta}(1-H_{s})\zeta dH_{s}\\
&-\sum_{i=1}^{N_{t}}\min(e^{-rT_i}X_{i},e^{-rT_i}a^*(t)).\end{aligned}\nneq

For the case $H_{t}=0$, we have
\begin{equation*}
J(s,Y_{s}^{\pi^{*}}-m^{*}(s)\zeta,Z_{s},1)=-\frac{h^{P}}{\Delta}\tilde{\xi }%
(s,Z_{s})\exp\{-\alpha Y_{s}^{\pi^{*}}e^{r(T-s)}\}
\end{equation*}
\begin{equation*}
J(s,Y_{s}^{\pi^{*}},Z_{s},0)=-\tilde{\xi}(s,Z_{s})\exp\{-\alpha Y_{s}^{\pi
^{*}}e^{r(T-s)}\}
\end{equation*}
Then, we need only obtain an estimate of:
\begin{equation*}
\mathbb{E}\bigg[J^{2}(s,Y_{s}^{\pi^{*}}-m^{*}(s)\zeta,Z_{s},1)\bigg]\newline
=\bigg(\frac{h^{P}}{\Delta}\bigg)^{2}\mathbb{E}\bigg[\tilde{\xi}%
^{2}(s,Z_{s})\exp\bigg\{-2\alpha Y_{s}^{\pi^{*}}e^{2r(T-s)}\bigg\}\bigg]
\end{equation*}
and
\begin{equation*}
\mathbb{E}\bigg[J^{2}(s,Y_{s}^{\pi^{*}},Z_{s},0)\bigg]\newline
=\mathbb{E}\bigg[\tilde{\xi}^{2}(s,Z_{s})\exp\bigg\{-2\alpha
Y_{s}^{\pi^{*}}e^{2r(T-s)}\bigg\}\bigg]
\end{equation*}
by the same argument in Step 1 in the proof of Theorem \ref{Post-Strategy},
we can get the result. Similarly, we have the same result for the case $%
H_{t}=1$. Then by Corollary 7.8 in \cite{KFC}, we conclude that
$$
J({\tau_{i}\wedge T}, Y^{\pi^{*}}_{\tau_{i}\wedge T}-m^{*}(\tau_{i}\wedge
T)\zeta, Z_{\tau_{i}\wedge T}, 1-H_{\tau_{i}\wedge T})-J({\tau_{i}\wedge T},
Y^{\pi^{*}}_{\tau_{i}\wedge T}, Z_{\tau_{i}\wedge T}, H_{\tau_{i}\wedge T})
$$
is uniformly integrable.
\end{proof}

\subsection{Numerical results}
In this section, we solve the Cauchy problem (3.16) and the first initial-boundary value problem (3.40) by using the finite-difference method. First, we assume that the claims are exponentially distributed with parameter $b$, and $T<\frac{1}{r}\log(b/\alpha)$,.the first step is to reduce the problem (3.16) and (3.42) to a bounded domain, i.e., $\mathbb{R}$ is replaced by $[-d,d],d<\infty$, and to add artificial boundary conditions. Then the Cauchy problem (3.16) to solve is the following:\\

\beq
\left\{\begin{aligned}
&0=\xi_t+\frac{1}{2}\beta^2\xi_{zz}+g(z)\xi_z-\xi\bigg\{\frac{(\mu(z)-r)^2}{2\sigma^2(z)}
+\alpha e^{r(T-t)}\bigg[(1+\eta)\lambda\mu_\infty\\
&-\frac{\lambda}{b}\exp\{(1-\frac{b}{\alpha}e^{-r(T-t)}\ln(1+\theta))\}\bigg]-\lambda\alpha\frac{e^{r(T-t)}}{\alpha e^{r(T-t)}-b}[\exp\{(\alpha e^{r(T-t)}-b)\frac{e^{-r(T-t)}}{\alpha}\ln(1+\theta)\}-1]\bigg\},\\
&\xi(z,T)=1,  \forall z\in[-d,d],\\
&\xi(z,t)=1,  \forall z\bar{\in}]-d,d[\times[0,T].
\end{aligned}\right.\end{equation}

From Friedman(1975), we know that the solution of (3.34) exists and is unique. The imposed boundary conditions give a good error estimate for large values of $d$.\\

Now we discretize (3.43) in the domain $A:=[-a,a]\times[0,T]$. A uniform grid on $A$ is given by:\\

\begin{align*}
&z_i=-d+(i-1)h, i=1,...,N, h=2d/(N-1),\\
&t_j=(j-1)k,  j=1,...,M,  k=T/(M-1).
\end{align*}

The space and time derivatives are discretized using finite differences as follows:\\
\begin{align*}
&\xi_t(z_i,t_j)\simeq\frac{\xi(z_i,t_j)-\xi(z_i,t_j-k)}{k},\\
&\xi_z(z_i,t_j)\simeq\frac{\xi(z_i+h,t_j)-\xi(z_i-h,t_j)}{2h},\\
&\xi_{zz}(z_i,t_j)\simeq\frac{\xi(z_i+h,t_j)-2\xi(z_i,t_j)+\xi(z_i-h,t_j)}{h^2}.
\end{align*}

We denote by $\xi_i^j:=\xi(z_i,t_j)$ the solution on the discretized domain. Then by substituting the derivatives by the expressions given above, (3.34) becomes:\\

\begin{align*}
\frac{\xi_i^j-\xi_i^{j-1}}{k}&+\frac{1}{2}\beta^2\frac{\xi_{i+1}^j-2\xi_i^j+\xi_{i-1}^j}{h^2}
+g(z)\frac{\xi_{i+1}^j-\xi_{i-1}^j}{2h}-\xi_i^j\{\frac{(\mu(z_i)-r)^2}{2\sigma^2(z_i)}\\
&+\alpha e^{r(T-t_j)}[(1+\eta)\lambda\mu_\infty
-\frac{\lambda}{b}\exp\{(1-\frac{b}{\alpha}e^{-r(T-t_j)}\ln(1+\theta))\}]\\
&-\lambda\alpha\frac{e^{r(T-t_j)}}{\alpha e^{r(T-t_j)}-b}[\exp\{(\alpha e^{r(T-t_j)}-b)\frac{e^{-r(T-t_j)}}{\alpha}\ln(1+\theta)\}-1]\}=0.
\end{align*}

Then for $i=2,...,N-1$ and $j=2,...,M$, $\xi_i^j$ satisfies the following explicit scheme:\\
\beq
\begin{aligned}
\xi_i^{j-1}=&(1-\frac{k\beta^2}{h^2}-k(\frac{(\mu(z_i)-r)^2}{2\sigma^2(z_i)}+\alpha e^{r(T-t_j)}[(1+\eta)\lambda\mu_\infty
-\frac{\lambda}{b}\exp\{(1-\frac{b}{\alpha}e^{-r(T-t_j)}\ln(1+\theta))\}]\\
&-\lambda\alpha\frac{e^{r(T-t_j)}}{\alpha e^{r(T-t_j)}-b}[\exp\{(\alpha e^{r(T-t_j)}-b)\frac{e^{-r(T-t_j)}}{\alpha}\ln(1+\theta)\}-1]))\xi_i^j\\
&+(\frac{k\beta^2}{2h^2}+\frac{k}{2h}g(z_i))\xi_{i+1}^j
+(\frac{k\beta^2}{2h^2}-\frac{k}{2h}g(z_i))\xi_{i-1}^j.
\end{aligned}
\nneq

The final condition is given by:\\
$\xi_i^M=1$, for all $i=1,...,N$.\\
The imposed boundary conditions will be given by:\\
$\xi_1^j=1$, for all $j=1,...,M-1$,\\
$\xi_{N+1}^j=1$, for all $j=1,...,M-1$.\\

Similarly, we can obtain $u_i^j$ satisfies the following explicit scheme:\\
\beq
\begin{aligned}
u_i^{j-1}=&(1-\frac{k\beta^2}{h^2}-\frac{kh^p}{\Delta}u_i^j
+(\frac{k\beta^2}{2h^2}+g(z_i)\frac{k}{2h})u_{i+1}^j
+(\frac{k\beta^2}{2h^2}-g(z_i)\frac{k}{2h})u_{i-1}^j\\
&-k\{\frac{(\mu(z_i)-r)^2}{2\sigma^2(z_i)}+\alpha e^{r(T-t_j)}[(1+\eta)\lambda\mu_\infty
-\frac{\lambda}{b}\exp\{(1-\frac{b}{\alpha}e^{-r(T-t_j)}\ln(1+\theta))\}]\\
&-\lambda\alpha\frac{e^{r(T-t_j)}}{\alpha e^{r(T-t_j)}-b}[\exp\{(\alpha e^{r(T-t_j)}-b)\frac{e^{-r(T-t_j)}}{\alpha}\ln(1+\theta)\}-1]\\
&+(1-\frac{1}{\Delta}+\frac{1}{\Delta}\ln\frac{1}{\Delta})h^p-\frac{h^p}{\Delta}\ln\xi_i^{j-1}\}.
\end{aligned}
\nneq

and we have
\begin{align}
\bar{\xi}_i^j=\exp\{u_i^j\}.
\end{align}

The final condition is given by:
$u_i^M=0$, for all $i=1,...,N$.
The imposed boundary conditions will be given by:
$u_1^j=0$, for all $j=1,...,M-1$,
$u_{N+1}^j=0$, for all $j=1,...,M-1$.

Our algorithm given by the explicit scheme, final condition and the imposed boundary conditions is backward in time, forward in space, and hence, by the explicit scheme, the numerical solution can be computed.

\bexa (The Value Functions)Suppose:
\begin{align*}
&r=0.04, \mu=0.3, \sigma(z)=e^z, \delta=0.1, \kappa=1, \lambda=3, \alpha=0.02,
\beta=0.3, b=2, d=2 ,\\
&T=5, \mu_\infty=1/2, \eta=7/3,\theta=8/3, h^p=0.25, \Delta=0.25, \zeta=0.4, N=401, M=50001.
\end{align*}
Harnessing the method (3.58), (3.59) and the relation (3.60), we can know the figures of assessment function before and after the cooperate bond default and conclusions as FIGURE 1.

\begin{figure}[htbp]
\subfigure[1]{
\begin{minipage}{8cm}
\centering
\includegraphics[width=8cm,height=5cm]{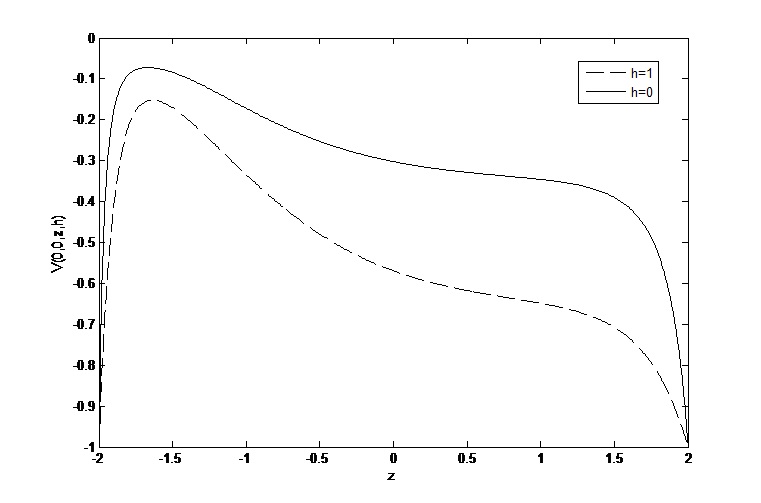}
\end{minipage}
}
\subfigure[2]{
\begin{minipage}{8cm}
\centering
\includegraphics[width=8cm,height=5cm]{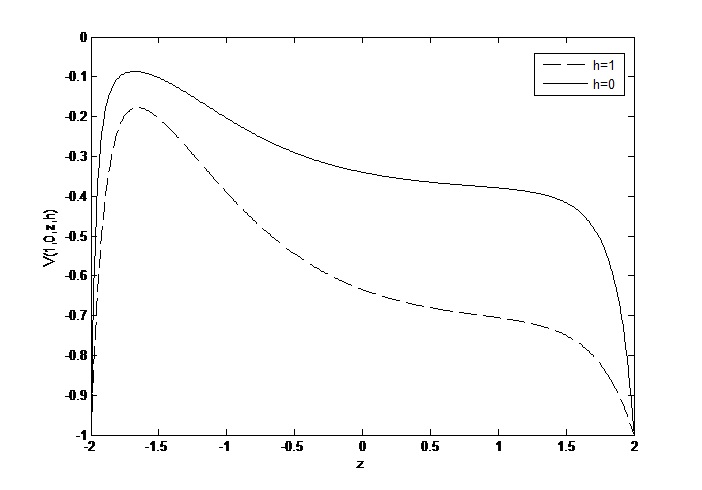}
\end{minipage}
}
\caption{}
\end{figure}

\obse From FIGURE 1, we conclude the following:
\begin{itemize}
 \item[(1)]	Assessing model is progressively decreasing by time $t$.
 \item[(2)]	Assessing procession is progressively increasing by $y$, which can be claimed by the function.
 \item[(3)] Before-defaulting assessing model is better than after-defaulting one obviously, which proves that Insurance companies can obtain much more profits after investing surplus in defaultable bonds.
\end{itemize}
\eobse
\begin{figure}[htbp]
\subfigure[1]{
\begin{minipage}{8cm}
\centering
\includegraphics[width=8cm,height=5cm]{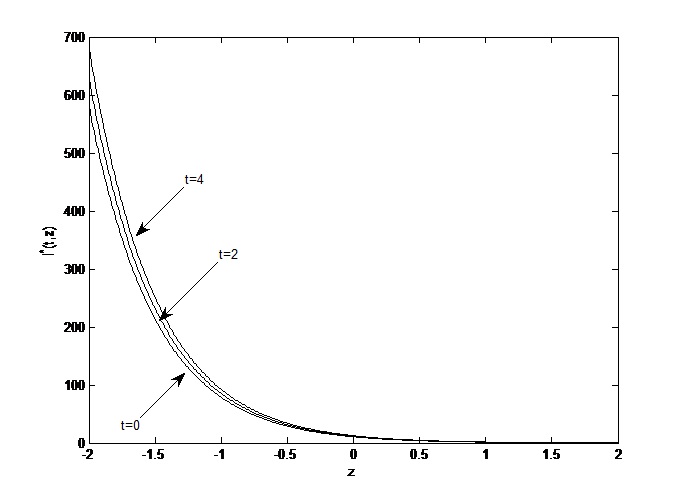}
\end{minipage}
}
\subfigure[2]{
\begin{minipage}{8cm}
\centering
\includegraphics[width=8cm,height=5cm]{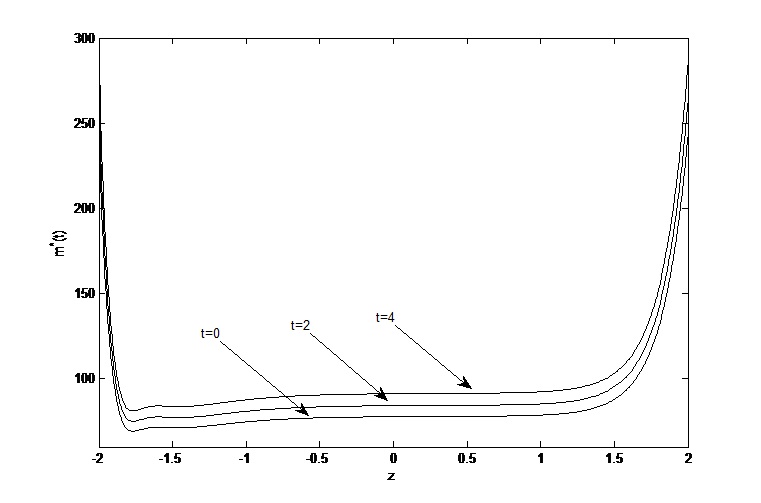}
\end{minipage}
}
\subfigure[3]{
\begin{minipage}{8cm}
\centering
\includegraphics[width=8cm,height=5cm]{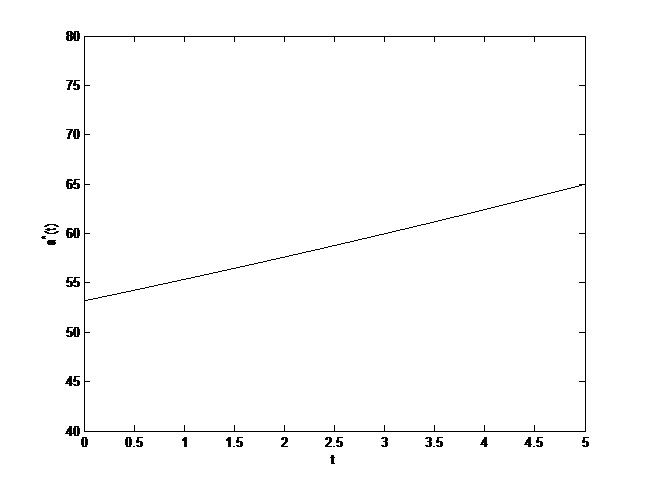}
\end{minipage}
}
\caption{}
\end{figure}
\obse The tendency of the optimal investing strategies $\pi^*(t)=(l^*(t),m^*(t),a^*(t))$ can be presented by FIGURE 2 respectively and the conclusions are followed:
\begin{itemize}
 \item[(1)] The investments in the asset of risk market is progressively decreasing in $z$ and increasing in $t$.
 \item[(2)] The investments in corporate bond is increasing in $t$. These will drop at first and then increase in $z$.
 \item[(3)]	The amount of retention of excess-of-loss reinsurance is increasing about $t$.
\end{itemize}
\eobse
\begin{figure}[htbp]
\subfigure[1]{
\begin{minipage}{6cm}
\centering
\includegraphics[width=6cm,height=4cm]{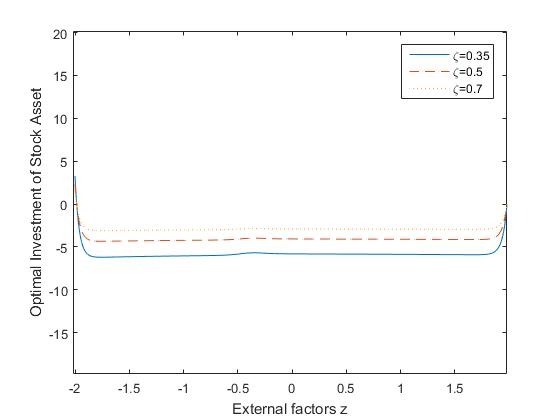}
\end{minipage}
}
\subfigure[2]{
\begin{minipage}{6cm}
\centering
\includegraphics[width=6cm,height=4cm]{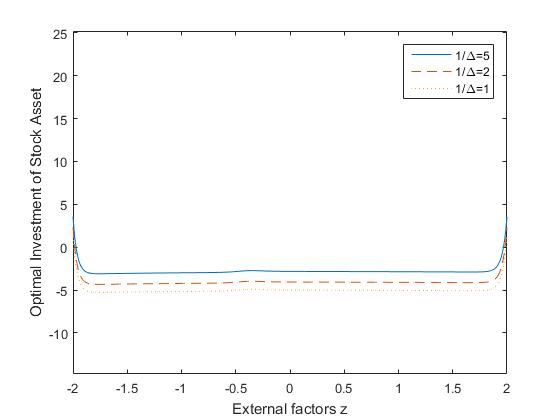}
\end{minipage}
}
\caption{The influence of the external factors of the optimal investment of stock asset with different loss rate and the default premium}
\end{figure}
\obse Considering the change of interval $\zeta$ and default risk premium $\frac{1}{\Delta}$, we need to make deeper numerical analysis.
\begin{itemize}
 \item[(1)]	In FIGURE 3(a)1, the external factor leads to the decrease of the optimal strategy  at first and then the increase. At the same time, the corporate bond is positively correlated with default risk premium $\frac{1}{\Delta}$. Insurance companies should invest a larger proportion of asset on corporate bond with higher risk of default.

 \item[(2)] In FIGURE 3(b)2, the insurer companies will introduce fewer investment in corporate bond when the loss rate is lower. In a nutshell, the adding $\zeta$ reflects few influence on the optimal investment of a corporate bond.
\end{itemize}
\eobse
\nexa

\bexa
Suppose:
\begin{align*}
&r=0.04, \mu=0.3, \sigma(z)=e^z, \delta=0.1, \kappa=1, \lambda=3, \alpha=0.2,
\beta=0.3, b=2, d=2 ,\\
&T=50, \mu_\infty=1/2, \eta=7/3,\theta=8/3, h^p=0.25, \Delta=0.25, \zeta=0.4, N=401, M=50001.
\end{align*}
\begin{figure}[htbp]
\subfigure[1]{
\begin{minipage}{6cm}
\centering
\includegraphics[width=6cm,height=4cm]{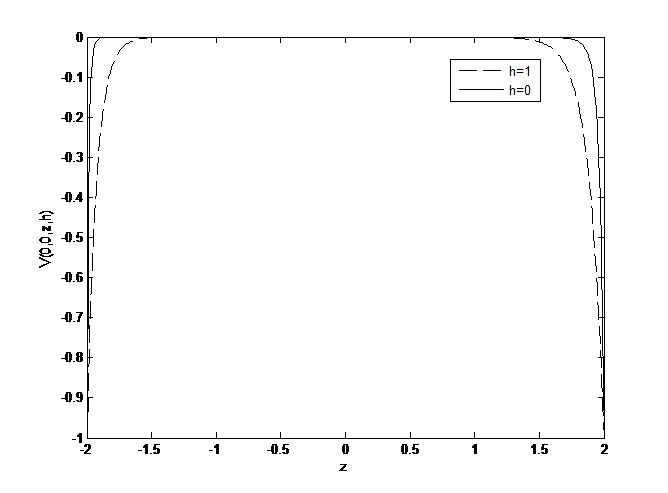}
\end{minipage}
}
\subfigure[2]{
\begin{minipage}{8cm}
\centering
\includegraphics[width=6cm,height=4cm]{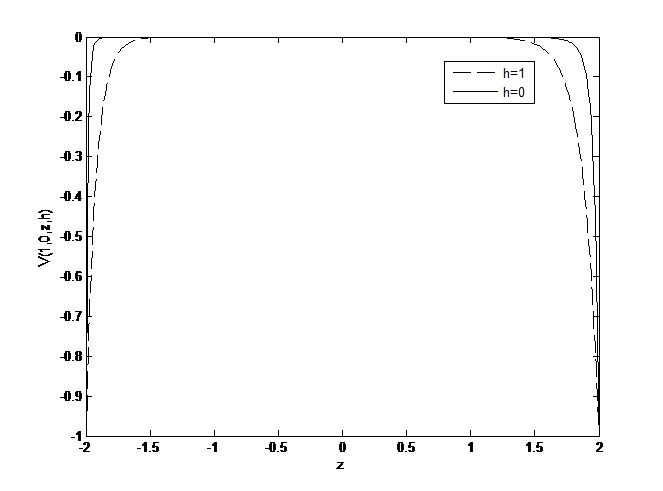}
\end{minipage}
}
\caption{}
\end{figure}
\obse
The FIGURE 4 express the situation of before default and after default. In pictures, the insurance companies can put most money on defaultable cooperate bond for more profit.
\eobse
\nexa

\bexa(The Sensitivity of the Optimal Investment of a Corporate Bond)
Assume $T-t=1$, $\alpha=0.5$, $r=0.04$. Then we operate the optimal strategy for $\frac{1}{\Delta }\in  [1,10]$ and $\theta \in [0.1,1]$. Firstly, fixing varying parameter $\xi$, we make comparisons between different $\zeta$ and $1/\Delta$. The function of the corporate bond can be expressed as follow:
\beq
m^*(t)=\frac{ln\frac{1}{\Delta }}{\alpha \zeta}.
\nneq
The comparisons were presented by following FIGURE 5.
\begin{figure}[htbp]
  \centering
  \subfigure[1]{
  \begin{minipage}{6cm}
  \centering
  \includegraphics[width=6cm,height=4cm]{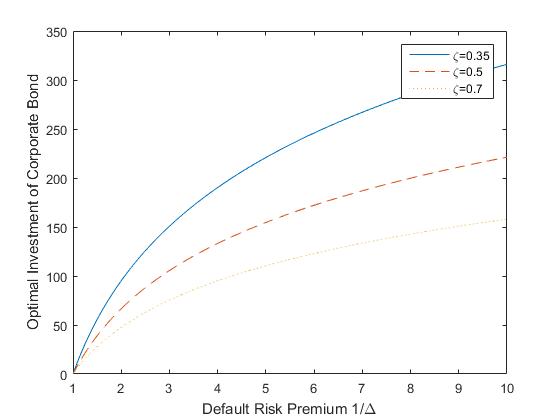}
  \end{minipage}
  }
  \subfigure[2]{
  \begin{minipage}{6cm}
  \centering
  \includegraphics[width=6cm,height=4cm]{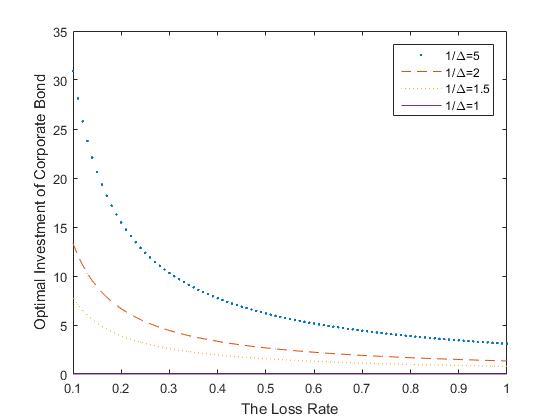}
  \end{minipage}
  }
  \subfigure[3]{
  \begin{minipage}{6cm}
  \centering
  \includegraphics[width=6cm,height=4cm]{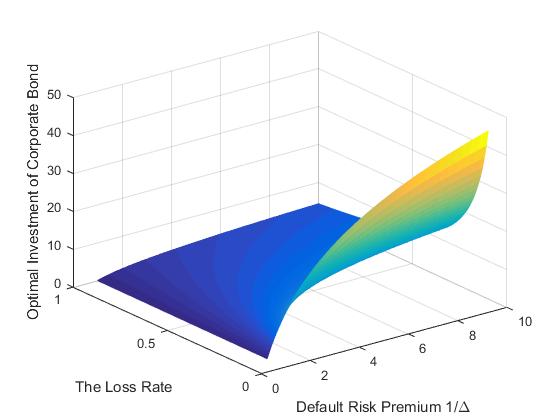}
  \end{minipage}
  }
  \caption{The influence of the loss rate and the default risk premium on the optimal investment of a corporate bond}
  \end{figure}

\obse Herein, we calculate the sensitivity of the optimal investment of a corporate bond. From FIGURE 5 we can tell:
\begin{itemize}
 \item[(1)]	The optimal investment of corporate for the default risk has positive relationship with default risk premium in FIGURE 5(a)1. The insurance companies will invest a relatively amount of money in a corporate bond with higher default risk condition.
 \item[(2)] There is a negative relation between loss rate and the optimal investment in FIGURE 5(b)2. FIGURE 5(b)2 describes that insurer will reduce the investment in corporate bond with increasing loss rate.
 \item[(3)] If the risk premium satisfies $ \frac{1}{\Delta }=1$, the insurance companies will not invest in corporate bond any more.  FIGURE 5(c)3 depicted comprehensive result.
\end{itemize}
\eobse
\nexa

\bexa(The Effect of RAP on OPRS)
When we treat $T=10$, $r=0.04$, the analysis of reinsurance strategy can be explained by exponential value function factor $\alpha$. Now we have $t\in[0,T]$, which means $t\in [0.10]$. We adopted various parameters $\alpha $ in order to compare the effectiveness of optimal excess-of-loss reinsurance.
Now the optimal excess-of-loss reinsurance was expressed:
\beq
\alpha^*(t)=\frac{ln(1+\theta)}{\alpha}e^{-r(T-t)}.
\nneq
According to the preconditions, the results can be told by FIGURE 6
\begin{figure}[htbp]
\centering
\subfigure[1]{
\begin{minipage}{6cm}
\includegraphics[width=6cm,height=4cm]{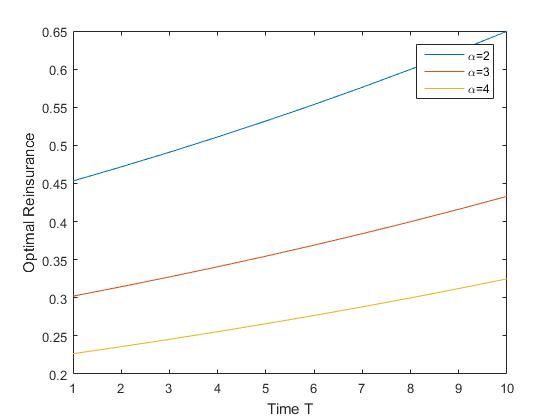}
\end{minipage}
}
\subfigure[2]{
\begin{minipage}{6cm}
\centering
\includegraphics[width=6cm,height=4cm]{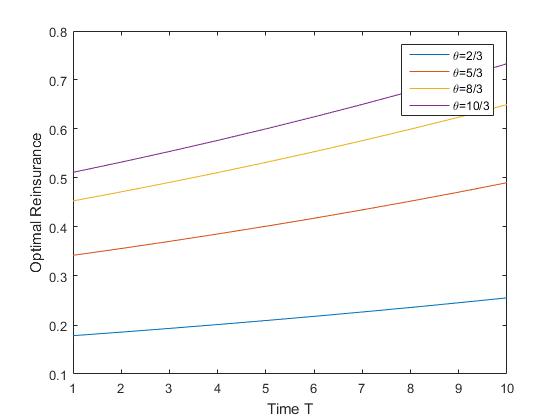}
\end{minipage}
}
\caption{The influence of the insurer's RAP on OPR}
\end{figure}
\obse According to FIGURE 6, the conclusions are presented below:
\begin{itemize}
 \item[(1)] From FIGURE 6(a)1, The utility of optimal excess reinsurance is increasing in time $t$.
 \item[(2)] When the parameter $\alpha $ grows progressively, the effect of the optimal investment is limited. The insurers will be willing to purchase more excess-of-loss reinsurance in order to reduce the risk of investing a value function with higher interval.
 \item[(3)] We can compare the safety loading $sigma$. Varying safety loading $sigma$ can generate multiple effects of reinsurance strategies, which can be compared by using previous data.
 \item[(4)] From FIGURE 6(b)2, when $sigma$ is bigger, the utility of excess-of-loss reinsurance strategies will be larger. If the insurers purchase the investing products with higher parameter $sigma$, they will need to restrain this kind of investment. In contrast, the companies should invest more money on a strategy with lower $sigma$.
\end{itemize}
\eobse
\nexa

\bexa(The Effect of RAP on OPRS)
The aim of discussion is the relationship between property and exponential value function factor $\alpha$. Suppose $T=10$, $r=0.04$, and then $t\in[0,T]$, which means $t\in[0.10]$.The relation of property $l^*(t)$ is:

\beq
l^*(t)=\frac{\mu(z)-r}{\alpha \sigma ^2(z)}e^{-r(T-t)}.
\nneq

In this function, the volatility $\sigma(z)$ is
\beq
\sigma(z)=e^z,
z_{i}=-a+(i-1)h,
h=\frac{2a}{n-1}.
\nneq
From these functions, we can make further assumption  $a=2$ and $n=10$, and then the results are shown as FIGURE 7:
\begin{figure}[H]
\centering
\subfigure[1]{
\begin{minipage}{6cm}
\centering
\includegraphics[width=6cm,height=4cm]{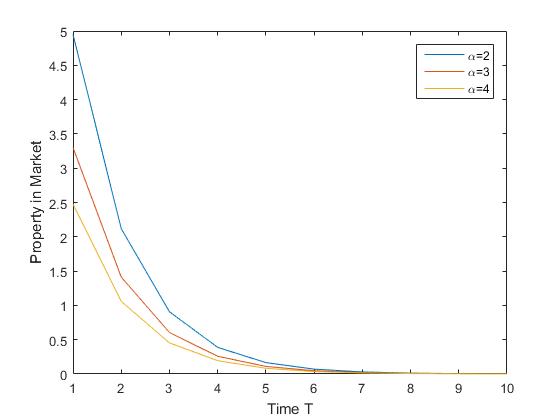}
\end{minipage}
}
\subfigure[2]{
\begin{minipage}{6cm}
\centering
\includegraphics[width=6cm,height=3cm]{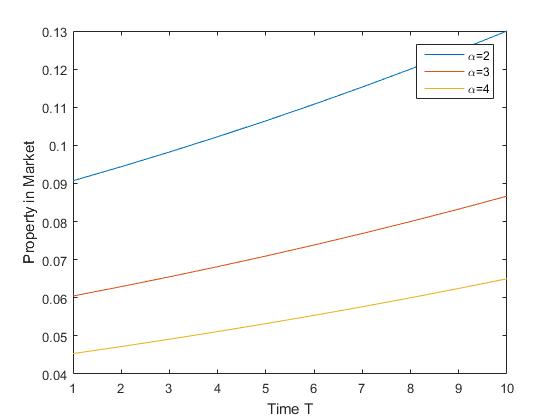}
\end{minipage}
}
\caption{The influence of the insurer's RAP on OPR}
\end{figure}

\obse According to FIGURE 7, the conclusions are as follows:
\begin{itemize}

 \item[(1)] the longer length of time will result in less utility of investments in market. When the parameter is increasing, the property in investment will reduce. Consequently, for insurers, the investment in large factor $\alpha $ will bring about restricted fortune in risky market.
 \item[(2)] If we ignore the volatility of the market or treat all volatility are the same, the result are as what our FIGURE 7(b)2 about. The results are totally different between the result which have same volatility or not.
 \item[(3)] Now more and more money are invested in market with increasing time. Obviously, if adding the consideration of volatility, the insurers will put less property on market in longer time. As a result, the longer time will generate lager volatility, larger uncertain factors and larger risk. In order to obtain steady income, we do not need to invest more money on market later. However, when the factor $\alpha $ increases, the money which put on market will reduce. If insurers decide to focus on an investment in value function with a lager parameter, the market property will reduce.
\end{itemize}
\eobse
\nexa

\section{Acknowledgments}

The authors would like to thank Professor Lijun Bo, for his detailed guidance and instructive suggestions.
N. Yao was supported by Natural Science Foundation of China (11101313, 113713283).

\section{Appendix}

\begin{thm}
\label{A.1} (Friedman,1975). We consider the following Cauchy problem
\begin{equation}  \label{F-Cauchy}
\left\{\begin{aligned} &u_t(x,t)+\mathcal L u(x,t)=f(x,t) \quad in \quad
\mathbb{R}^n\times[0,T)\\ &u(x,T)=h(x) \quad in \quad \mathbb{R}^n,
\end{aligned}\right.
\end{equation}

Where $\mathcal{L}$ is given by:
\begin{equation*}
\mathcal{L }u=\frac{1}{2}\sum_{i,j=1}^na_{ij}(x,t)u_{x_ix_j}+%
\sum_{i=1}^nb_i(x,t)u_{x_i}+c(x,t)u.
\end{equation*}

If the Cauchy problem (\ref{F-Cauchy}) satisfies the following conditions:

\begin{itemize}
\item[1.] The coefficients of $\mathcal{L}$ are uniformly elliptic;

\item[2.] The functions $a_{ij}$, $b_i$ are bounded in $\mathbb{R}^n\times[%
0,T]$ and uniformly Lipschitz continuous in $(x,t)$ in compact subsets of $%
\mathbb{R}^n\times[0,T]$;

\item[3.] The functions $a_{ij}$ are H$\ddot{o}$lder continuous in $x$,
uniformly with respect to $(x,t)$ in $\mathbb{R}^n\times[0,T]$;

\item[4.] The function $c(x,t)$ is bounded in $\mathbb{R}^n\times[0,T]$ and
uniformly H$\ddot{o}$lder continuous in $(x,t)$ in compact subsets of $%
\mathbb{R}^n\times[0,T]$;

\item[5.] $f(x,t)$ is continuous in $\mathbb{R}^n\times[0,T]$, uniformly H$%
\ddot{o}$lder continuous in $x$ with respect to $(x,t)$ and $|f(x,t)|\leq
B(1+|x|^\gamma)$;

\item[6.] $h(x)$ is continuous in $\mathbb{R}^n$ and $|h(x)|\leq B(1 +
|x|^\gamma)$, with $\gamma> 0$;
\end{itemize}

then there exists a unique solution $u$ of the Cauchy problem (4.1)
satisfying:
\begin{equation*}
|u(x,t)| \leq const(1 + |x|^\gamma) \quad \quad and \quad \quad |u_x(x, t)|
\leq const(1 + |x|^\gamma).
\end{equation*}
\end{thm}

\remarks In the original theorem of Friedman(1975), the Cauchy problem is
given by
\begin{equation}
\left\{\begin{aligned} &v_t(x,t)-\mathcal L v(x,t)=f(x,t) \quad in \quad
\mathbb{R}^n\times[0,T)\\ &v(x,0)=h(x) \quad in \quad \mathbb{R}^n,
\end{aligned}\right.
\end{equation}
We let $v(x,T-t)=u(x,t)$ , then we can get the Cauchy problem shown in (\ref%
{F-Cauchy}).

\end{document}